%% file: core.tex
\documentclass[11pt]{article}

\usepackage{amssymb,amsthm,amsmath,amssymb,wrapfig,dsfont,authblk}
\usepackage{thmtools,thm-restate}
\PassOptionsToPackage{hyphens}{url}
\usepackage[square,sort,semicolon,numbers]{natbib}
\usepackage{varioref}
\usepackage[usenames,svgnames,xcdraw,table]{xcolor}
\definecolor{DarkBlue}{rgb}{0.1,0.1,0.5}
\definecolor{DarkGreen}{rgb}{0.1,0.5,0.1}
\usepackage[backref=page]{hyperref}
\hypersetup{
     colorlinks   = true,
     linkcolor    = DarkBlue, 
     urlcolor     = DarkBlue, 
	 citecolor    = DarkGreen 
}
\renewcommand*{\backref}[1]{}
\renewcommand*{\backrefalt}[4]{%
    \ifcase #1 (Not cited.)%
    \or        (Cited on page~#2)%
    \else      (Cited on pages~#2)%
    \fi}
\usepackage[capitalise,noabbrev]{cleveref}
\usepackage[margin=0.9in]{geometry}

\usepackage[linesnumbered,lined,boxed,ruled,vlined]{algorithm2e} 

\SetKwInOut{Parameter}{Parameter}
\SetAlFnt{\small}
\SetAlCapFnt{\small}
\SetAlCapNameFnt{\small}
\SetAlCapHSkip{0pt}
\IncMargin{-\parindent}

\usepackage{tikz}
\usetikzlibrary{fit,calc}

\colorlet{pink}{red!40}
\colorlet{blue}{cyan!60}
\colorlet{mygray}{gray!55}

\makeatletter
\renewcommand{\paragraph}{%
  \@startsection{paragraph}{4}%
  {\z@}{1.0ex \@plus 1ex \@minus .2ex}{-1em}%
  {\normalfont\normalsize\bfseries}%
}
\let\oldnl\nl
\newcommand{\nonl}{\renewcommand{\nl}{\let\nl\oldnl}}
\makeatother

\usepackage{nicefrac}
\usepackage{enumerate}
\usepackage[labelfont={normalfont,bf},textfont=it]{caption}
\usepackage{subcaption}
\usepackage[T1]{fontenc}
\usepackage{mathtools}
\usepackage{bbm}

\usepackage{tabulary}
\usepackage{multirow}
\usepackage{booktabs}
\usepackage{colortbl}

\usepackage[backgroundcolor=blue!20!white]{todonotes}

\newtheorem{theorem}{Theorem}[section]

\newtheorem{lemma}[theorem]{Lemma}

\newtheorem{proposition}[theorem]{Proposition}

\newtheorem{definition}{Definition}[section]
\theoremstyle{definition}
\newtheorem{remark}[definition]{Remark}

\newcommand{\Econ}{\mathcal{E}}
\newcommand{\one}{\mathbf 1}
\newcommand{\zero}{\mathbf 0}

\def\ep{\varepsilon}
\def\one{\mathbf{1}}

\def\da{\delta}

\def\argmin{\mbox{argmin}}

\def\cvh{\mbox{cvh}}

\definecolor{Blue}{rgb}{0.,0.,1.}

\begin{document}

\title{{\bfseries The Edgeworth Conjecture with Small Coalitions and \\  Approximate Equilibria in Large Economies}}
\author{Siddharth Barman\thanks{Indian Institute of Science. \texttt{barman@iisc.ac.in}} \qquad  \qquad Federico Echenique\thanks{California Institute of Technology. \texttt{fede@hss.caltech.edu}}}

\date{}
\maketitle

\begin{abstract}We revisit the connection between bargaining and equilibrium in exchange economies, and study its algorithmic implications. We consider bargaining outcomes to be allocations that cannot be blocked (i.e., profitably re-traded) by coalitions of small size and show that these allocations must be approximate Walrasian equilibria. Our results imply that deciding whether an allocation is approximately Walrasian can be done in polynomial time, even in economies for which finding an equilibrium is known to be computationally hard.
\end{abstract}

\input{intro}

\input{results}

\input{notation}

\input{quant-core-convergence}

\input{testing-algorithm}

\section*{Acknowledgements}
Federico Echenique thanks the National Science Foundation for its support through grants SES-1558757 and CNS-1518941. Siddharth Barman gratefully acknowledges the support of a Ramanujan Fellowship (SERB - {SB/S2/RJN-128/2015}) and a Pratiksha Trust Young Investigator Award.

\bibliographystyle{alpha}
\bibliography{core}


\input{appendix}

\end{document}

%% file: intro.tex
\section{Introduction}
We present a {\em quantitative core convergence theorem}, with economic and algorithmic implications. Economists have, since Francis Edgeworth in 1881, been interested in the convergence of bargaining in finite economies to market equilibrium. The celebrated {\em Edgeworth conjecture} states that bargaining indeterminacy and monopoly power disappear in large economies. Edgeworth focused on the contract curve, what we think of now as the core. The core, however, requires agents to join coalitions of arbitrary size. In our quantitative core convergence theorem we show that coalitions of a fixed size suffice. The fixed size depends polynomially on the approximation error, the number of goods in the economy, and consumers' heterogeneity. 

The economic analysis of markets is based on the notion of Walrasian (competitive) market equilibrium and price-taking behavior. Walrasian equilibrium requires that goods are assigned prices, each agent maximizes her utility subject to what she can afford, and the market clears. Walrasian equilibria provide a conceptually crisp framework for economic analysis: with prices as guide, agents direct themselves, in a decentralized manner, to a Pareto optimal allocation. That agents are price takers means that they treat prices as fixed; they will not renegotiate the terms of trade with other agents. 

The Edgeworth conjecture, or core convergence, is the economist's basic justification for the price-taking assumption. Specifically, it postulates that, in a large economy, bargaining collapses to Walrasian equilibrium, i.e., in a large economy, no group of agents would choose to upset an equilibrium by renegotiating among themselves.


In a small economy, agents bargain over the terms of trade. They may not trade at fixed prices. Each individual agent in a small economy is to some degree unique and can therefore command a certain degree of monopoly power. Imagine the market for professional soccer players, or high-profile academics, in which each individual agent is special. Then terms of trade are determined by haggling and bargaining, not as fixed prices in competitive markets. By contrast, in a large economy with limited heterogeneity, no agent is special. Imagine the market for a commodity, like corn or steel, or for a standard piece of technology, like a personal desktop. In such economies, bargaining outcomes (formally, outcomes in the core) should collapse to the competitive outcomes that correspond to a Walrasian equilibrium. This claim is the Edgeworth conjecture.

Bargaining is a complex process, and difficult---arguably impossible---to model in detail. So economists use the notion of the core to impose discipline on bargaining outcomes. The core of an economy is the set of allocations (redistribution of goods among agents) that no subset, or coalition, of agents can improve upon via re-trading among themselves. In other words, the core comprises the outcomes that are not ``blocked'' by any coalition of agents. Any allocation that is not in the core would be renegotiated, or re-contracted, by some set of agents. Note that the definition of the core involves all possible coalitions of agents, of all possible sizes.

The core convergence theorem states that in a large (in an asymptotic sense) economy any outcome in the core, arrived at through seemingly indeterminate reallocations and bargaining, will, in fact, be a Walrasian equilibrium. In other words, the core convergence theorem captures the meaning of the Edgeworth conjecture, and supports the use of Walrasian equilibria and price-taking behavior in large economies. 

Edgeworth's conjecture was first formalized by Debreu and Scarf~\cite{debreu1963limit} as well as Aumann~\cite{aumann1964markets}.\footnote{See \cite{hildenbrand2015core} and \cite{anderson1992core} for reviews of the literature on core convergence.} In our paper we shall follow Debreu and Scarf; Aumann's model assumes a limit economy with a continuum of agents, hence there is no scope in his model to address the questions we are interested in.\footnote{A ``fractional'' version in atomless economies is still possible, see \cite{schmeidler1972remark,grodal1972second,vind1972third}.} Debreu and Scarf consider a {\em large economy with limited heterogeneity.} Their idea, which goes back to Edgeworth himself, is to postulate a finite number of agent types, say $h$, and then imagine a sequence of replica economies. The $n$-th replica of the economy has $n$ identical copies of each type of agent $t \in [h]$. Under standard assumptions, they show that, in the limit as $n\rightarrow\infty$, the core of the $n$th replica economy approaches the set of Walrasian equilibrium allocations.
Note that, the core convergence theorem not only requires $n$ to be asymptotically large, but, to address the core, one has to account for coalitions that are arbitrarily large, both in size and in number.  

In contrast to the core, we focus on the $\kappa$-core: the set of allocations that cannot be blocked by any coalition of size at most $\kappa$. In a large economy, it is unlikely that coalitions of arbitrary size can function effectively. Even if a large coalition can effectively communicate among its members, they would face significant hurdles in aggregating their preferences to achieve an outcome that is collectively better for the members of the coalition. Think of Arrow's theorem: preference aggregation is known to be very problematic, and the size and number of all coalitions in a large economy makes the problem even worse. The $\kappa$-core demands much less of the blocking coalitions, and seems to us a more plausible solution concept than the core. In addition, the number of coalitions to check for is much smaller in the $\kappa$-core than in the core (${n \choose \kappa}$ with fixed $\kappa$, rather than $2^n$). 

Our main result is that under benign assumptions on agents' preferences, an allocation that is in the $\kappa$-core must be an approximate equilibrium (an $\ep$-Walrasian equilibrium).  $\kappa$ depends polynomially on the approximation guarantee and on the agent heterogeneity in the economy. Our result is non-asymptotic; it  holds for finite economies, as long as the number of agents, $n$, exceeds $\kappa$.  We adopt the model and the assumptions of Debreu and Scarf. In fact, to make our paper self-contained, their result is restated in Section~\ref{section:quant-debreu-scarf}. 

We are not the first to prove a quantitative version of Edgeworth's conjecture. The closest paper to ours is Mas-Colell \cite{mas1979refinement}, who also exhibits a bound on the size of blocking coalitions needed for core convergence.\footnote{Mas-Colell's result builds on techniques due to Anderson  \cite{anderson1978elementary}.  Anderson's theorem, however, does not bound the size of the blocking coalitions, and hence does not speak to the $\kappa$-core. His result connects the size (and other parameters) of the economy with the approximation guarantee.} The main difference between Mas-Colell's result and ours is that he considers an average error as his approximation guarantee. He ensures that the average budget gap, where the average is taken over all agents in the economy, is small. If we translate his results into our stronger objective of a {\em per-agent} approximation guarantee, his result requires  $\kappa$ to be a function of $n$.\footnote{Per-agent approximation guarantees are the focus of modern results in algorithmic game theory. In the terminology of Anderson's survey \cite{anderson1992core} it reflects uniform convergence (what he terms U3), and it is rare in the older literature on core convergence, even for results that rely on the power of all coalitions.} Our result will instead say that $\kappa$ is $\mathcal{O}(1/\ep^2)$ (as long as $n\geq \kappa$).  We should emphasize that Mas-Colell's theorem requires substantially weaker assumptions than ours and does not rely on the Edgeworth-Debreu-Scarf replica setting. 

Our main theorem has important algorithmic consequences for the problem of deciding whether an allocation can be supported as an approximate Walrasian equilibrium.  Imagine a policy maker who would like to install a given allocation using decentralized markets. The standard advise to such a policy maker comes from the second welfare theorem, which says that any Pareto optimal allocation can be obtained (given certain assumptions on the economy) as a Walrasian equilibrium, using a system of taxes and subsidies to adjust all agents' incomes. We study the same question as in the second welfare theorem, but when the policy maker is unable to freely use taxes or subsidies. 

Specifically, given an allocation in an exchange economy, we consider the problem of deciding whether there exists prices that make the allocation an approximate Walrasian equilibrium. Our main algorithmic result is to provide a polynomial-time algorithm that resolves the question. Our finding stands in stark contrast to the existing hardness results for Walrasian equilibria, see, for example, \cite{chen09hardness,codenotti2006leontief,vazirani10,garg2017settling}. There are settings for which it is known that finding an equilibrium is hard, and for which our algorithm efficiently decides  whether a given allocation can be supported as an approximate Walrasian equilibrium.

We are, to the best of our knowledge, the first to study the question of ``testing'' a putative equilibrium allocation. That is, without use of taxes and subsidies, and thus outside of the second welfare theorem framework. Our paper resolves this question with a satisfactory algorithm. 

%% file: results.tex
\subsection{Informal statement of our results}

Our paper makes an economic and an algorithmic contribution. The economic contribution is in the form of a non-asymptotic core convergence theorem, for a weakened notion of the core: one that only requires robustness to blocks by small coalitions. The algorithmic contribution is to the problem of testing whether an allocation is a (approximate) Walrasian equilibrium allocation. The main lemma in our core convergence result allows us to develop an efficient algorithm that decides whether an allocation can be supported as an approximate Walrasian equilibrium (and find the supporting/equilibrium prices, if they exist). 

The $\kappa$-core is the set of allocations that are not blocked by coalitions of size at most $\kappa$. We show that core convergence is obtained as long as $\kappa$ is polynomially large (in the approximation parameter and agent heterogeneity) and the size of the economy is at least $\kappa$. Thus, core convergence is obtained in finite economies. 

Specifically, we prove that in any replica economy with well-behaved utilities, allocations in the $\kappa$-core are, in fact, $\varepsilon$-Walrasian, as long as $\kappa$ is $\mathcal{O}\left( \frac{h^2 \ell}{\varepsilon^2}\right)$; here, $h$ denotes the heterogeneity of the economy (i.e., the number of different types of agents), $\ell$ denotes the number of goods, and $\varepsilon>0$ is the approximation parameter; see Theorem~\ref{thm:quant-debreu-scarf}.  Note that our result does not require the number of agents $n$, or the size of blocking coalitions, to be arbitrarily large. The result is applicable whenever $n$ is $\Omega(\kappa)$. 

Our work develops new techniques for addressing core convergence. Instead of relying on limiting, or measure-theoretic, arguments as in the classical literature (see~\cite{hildenbrand2015core}), our proof builds upon geometric insights and dimension-free results. Specifically, we employ the approximate version of Carath\'{e}odory's theorem (see~\cite{barman2015approximating} and references therein). The classical theorem of Debreu and Scarf requires a ``rounding to rational numbers'' argument--the size of the economy needs to be arbitrarily large so that such a rounding is possible. Our proof avoids such a rounding by invoking the approximate Carath\'{e}odory theorem. In this manner, we ascertain that the largest coalitions needed to ensure an approximate Walrasian equilibrium are only polynomially large and independent of the size of the economy. 

Our key technical insight is provided through an efficiently implementable characterization of approximate Walrasian equilibrium allocations (see Lemma~\ref{lemma:bounded-hull} in Section~\ref{section:bounded-hull}). Our characterization yields, not only a quantitive version of core convergence, but also an efficient algorithm that tests whether a given allocation can be supported as an approximate equilibrium.

Our algorithm takes as input an allocation, and determines whether there are prices that make the allocation into an approximate Walrasian equilibrium. This result is potentially useful to understand the policy objectives that can be obtained in a descentralized fashion, as a Walrasian equilibrium. It also contributes to the algorithmic literature on equilibrium computation, as it identifies a notable dichotomy between testing and computing Walrasian equilibria. 

The problem of computing Walrasian equilibria has been studied extensively in algorithmic game theory. Most results in this direction are negative: finding a Walrasian equilibrium is computationally hard for general settings; see, e.g.,~\cite{chen09hardness,codenotti2006leontief,vazirani10,garg2017settling}. By contrast, we show that the testing counterpart is computationally tractable. In particular, we develop a polynomial-time algorithm that, for a given allocation, finds prices (if they exist) that make the allocation an approximate Walrasian equilibrium (Theorem~\ref{theorem:testing-algorithm}).\footnote{Note that, in particular settings, it might be possible to efficiently test if a given allocation is Pareto optimal and, hence (using the characterization provided by the welfare theorems), determine whether the allocation can be supported as an \emph{exact} equilibrium. However, in and of itself, this observation does not extend to approximate equilibria. Furthermore, it does not directly lead to an efficient method for finding the supporting, equilibrium prices.}
 
The above-mentioned results require the utilities in the economy to be strongly concave. Hence, as is, they do not address piecewise-linear concave (PLC) utilities, which, while concave, are not strongly concave. The PLC case is particularly important in light of the hardness results in~\cite{deng2008computation, chen09hardness,garg2017settling}.  However, we show that the developed ideas can be adapted to obtain an efficient testing algorithm even for PLC economies (Theorem~\ref{theorem:plc}). 


%% file: notation.tex
\section{Notation and Preliminaries} 
\label{section:notation}
\noindent
{\bf Notational Conventions.} For vectors $x, y \in \mathbb{R}^\ell$, write $x \leq y$ iff $x_i \leq y_i$ for all $i \in [\ell]$. Also, $x \ll y$ denotes that $x_i < y_i$ for all $i \in [\ell]$. We will use $\Delta$ to refer to the standard simplex in $\mathbb{R}^\ell$, $\Delta \coloneqq \left\{ x \in \mathbb{R}^\ell \mid \sum_{i=1}^\ell x_i = 1 \text{  and  } x_i \geq 0 \text{ for all } i \right\}$. The convex hull of a set of vectors $A \in \mathbb{R}^\ell$ will be denoted by $\cvh(A)$. In addition, the all-ones and all-zeros vector will be denoted by $\one$ and $\zero$, respectively. \\

\noindent
{\bf Exchange economies.}  An exchange economy $\Econ$ comprises of a set of consumers, $[h] \coloneqq \{1, 2, \ldots, h\}$, and a set of goods, $[\ell] \coloneqq \{1, 2, \ldots, \ell \}$. Each consumer is endowed with different quantities of the goods as an endowment; specifically, for every consumer $i \in [h]$, the vector $\omega_i \in \mathbb{R}_+^\ell$ denotes (componentwise) the amount of each good endowed to $i$. The  preference of consumer $i \in [h]$, over bundles of goods, is specified by a utility function, $u_i : \mathbb{R}_+^\ell \mapsto \mathbb{R}$. In particular, every consumer is described by a pair $(u_i, \omega_i)$. An exchange economy $\Econ$ is a tuple $\left( (u_i, \omega_i) \right)_{i=1}^h$. 

We shall adopt some standard assumptions on consumers' utilities: $u_i$s are continuous and monotone increasing. We further assume that the utilities are continuously differentiable\footnote{This additional smoothness assumption is primarily for ease of exposition. One can extend the arguments developed in this work to continuous functions by replacing gradients with  subgradients.} and $\alpha$-strongly concave, with $\alpha >0$. Strong concavity (convexity) is a well-studied property in convex optimization (see, for example, \cite{boyd2004convex}). It provides a parametric strengthening of concavity. Formally, a differentiable function, $u: \mathbb{R}^\ell \mapsto \mathbb{R}$, is said to be $\alpha$-strongly concave within a set $\mathcal{R} \subset \mathbb{R}^\ell$ iff the following inequality holds for all $x, y \in \mathcal{R}$:
\begin{align*}
u(y) \leq u(x) + \nabla u(x)^T (y-x) - \frac{\alpha}{2} \| y - x \|^2.
\end{align*}
Here, $\nabla u(x)$ is the gradient of the function $u$ at point $x$ and $\| \cdot \|$ denotes the Euclidean norm. Note that if $\alpha = 0$, then $u$ is simply a concave function. Furthermore, the case of $\alpha>0$ corresponds to strict concavity. 

The set $\mathcal{R}$ specifies the subdomain over which strong concavity holds. We assume that $\mathcal{R}$ is appropriately large and, in particular, that it contains the Euclidean ball of radius $r \in \mathbb{R}_+$ and center $\omega_i$ (the endowment vector), for each $i$. We assume that the radius $r$ satisfies $\alpha r^2 \geq \frac{2\varepsilon \lambda \ell}{h} + 2$, with $\varepsilon >0$ being the approximation parameter and $\lambda$ denoting the Lipschitz constant of the utilities. At a high level, the condition asserts that agent $i$'s utility function has sufficient curvature close to $i$'s endowment $\omega_i$.  The strong concavity assumption is satisfied by standard examples of utility functions used in economics; in particular, utilities of the form $u(y) \coloneqq  ( \sum_{i=1}^\ell y_i^\rho)$, with $0 < \rho < 1$---i.e., utilities within the CES (constant elasticity of substitution) family of utilities---are strongly concave; see Appendix~\ref{appendix:example} for an illustrative example.  Throughout, for ease of presentation, we will simply say that the utilities are $\alpha$-strongly concave, with the set $\mathcal{R}$ being implicit. \\



\noindent
{\bf Core.} For an exchange economy with $h$ consumers and $\ell$ goods, $\Econ =  \left( (u_i, \omega_i) \right)_{i=1}^h$, we define the following central constructs 

\begin{itemize}
\item An \emph{allocation} in $\Econ$ is a vector, $\overline{x}=(\overline{x}_i)_{i=1}^h \in \mathbb{R}^{h\ell}_+$, such that $\sum_{i=1}^h \overline{x}_i=\sum_{i=1}^h \omega_i $. In other words, an allocation corresponds to a redistribution of the endowments among the consumers. 
\item A nonempty subset $S\subseteq [h]$ is a \emph{coalition}. Let $S$ be a coalition, then a vector $(y_i)_{i\in S}$ is an \emph{S-allocation} if $\sum_{i\in S} y_i=\sum_{i\in S} \omega_i$. 
\item A coalition $S$ \emph{blocks} the allocation $\overline{x} = (\overline{x}_i)_{i=1}^h$ in $\Econ$ if there exists an $S$-allocation $(y_i)_{i\in S}$ such that $u_i(y_i) > u(\overline{x}_i)$ for all $i\in S$. That is, a blocking coalition is a group of consumers that are better off trading among themselves, than they would have been under $\overline{x}$.  
\item The \emph{core} of $\Econ$ is the set of all allocations that are not blocked by any coalition.
\item The \emph{$\kappa$-core} of $\Econ$, for $\kappa\in \mathbb{Z}_+$, is the set of allocations that are not blocked by any coalition of cardinality at most $\kappa$. Hence, in an economy with $n$ consumers, the core is the $n$-core and the set of individually rational allocations are the $1$-core.
\end{itemize}


\noindent 
{\bf Equilibrium and approximate equilibrium.} 
In an exchange economy $\Econ= \left( (u_i, \omega_i) \right)_{i \in [h]}$, a \emph{Walrasian equilibrium} is a pair $(p, \overline{x}) \in \mathbb{R}^\ell_+ \times \mathbb{R}^{h \ell}_+$ in which  
\begin{enumerate}
\item $p \in \mathbb{R}_+^\ell$ is a price vector for the $\ell$ goods in the economy. 
\item $\overline{x}=\left(\overline{x}_i \right)_{i \in [h]} \in \mathbb{R}^{h \ell}_+$ is an allocation, i.e., under $\overline{x}$ supply equals the demand, $\sum_{i=1}^h \overline{x}_i = \sum_{i=1}^h \omega_i$.
\item Every consumer $i \in [h]$ maximizes its utility $u_i$, while consuming its endowment $\omega_i$. That is, $p^T \overline{x}_i = p^T \omega_i$ and, for all bundles $y \in \mathbb{R}^\ell_+$ with the property that $u_i(y) > u_i(\overline{x}_i)$, we have $p^T y_i > p^T \omega_i$. 

\end{enumerate}

Next, we describe what we mean by an approximate Walrasian equilibrium. We use a notion of approximation in which consumers are optimizing exactly, subject to budget constraints that are satisfied approximately. 

Formally, in an exchange economy $\Econ= \left( (u_i, \omega_i) \right)_{i \in [h]}$,  \emph{$\varepsilon$-Walrasian equilibrium} is a pair $(p, \overline{x}) \in \mathbb{R}^\ell_+ \times \mathbb{R}^{h \ell}_+$ in which the (normalized) price vector $p \in \Delta$ and the allocation $\overline{x} \in  \mathbb{R}^{h \ell}_+$ satisfy the following two conditions, for all consumers $i \in [h]$:  
\begin{itemize}
\item[(i)] $|p^T \overline{x}_i - p^T \omega_i| \leq \varepsilon$ and 
\item[(ii)] for any bundle $y \in \mathbb{R}^\ell_+$, with the property that $u_i(y) > u_i (\overline{x}_i)$, we have $p^T y > p^T \omega_i - \varepsilon/h$. 
\end{itemize}

Our notion of approximation is different from the standard use in algorithmic game theory, where agents are approximately optimizing, but it conforms to the most used definitions in economic theory: see~\cite{starr1969quasi,arrow1971general,hildenbrand1973existence,hildenbrand2015core,anderson1978elementary,andersonmktvalue,anderson1982approximate,mas1978note,mas1979refinement}. Economists do not view utilities as objective, observable, entities. Hence it is natural in economics to measure the approximation error in monetary terms, using consumers' expenditure, instead of measuring approximation in terms of utility loss. 

That said, it is straightforward to move between different notions of approximate equilibrium, and our result continues to apply. One can, instead, require that the approximation guarantee be in terms of the utilities. Or that budget exhaustion holds exactly (so the value  of agents' consumption equals the value of their endowment), but the overall supply is only approximately equal to the demand. In Appendix~\ref{appendix:other-notions} we show that one can obtain these kinds of approximations as well from the results in this present paper. \\


\noindent
{\bf Walrasian allocation.} The term \emph{(approximate) Walrasian allocation} refers to an allocation $\overline{x} \in \mathbb{R}^{h\ell}_+$ in $\Econ$ for which there exists a price vector $p \in \Delta$ such that $(p, \overline{x})$ is a (approximate) Walrasian equilibrium. \\
%

\noindent 
{\bf Utility normalization.} 
Our work addresses additive approximations, hence, we will follow the standard assumption (used in the context of absolute-error bounds) that the utilities are normalized so that $u_i(x_i) \in [0,1)$ for all consumers $i \in [h]$ and all allocations $(x_i)_i \in \mathbb{R}^{h \ell}_+$. In fact, we shall normalize then so that  $u_i(x_i) \in[0,1-\bar \eta)$, where $\bar\eta\in (0,1)$ for all consumers $i \in [h]$ and all allocations $(x_i)_i \in \mathbb{R}^{h \ell}_+$. These normalizations are possible as the set of all allocations is compact and, hence, each utility function will achieve a maximum over the space of allocations. 

Note that the utilities are normalized only for feasible allocations in the underlying exchange economy. A bundle which is, say, component-wise greater than the total endowment in the exchange economy---or, allocations in a replica economy (defined below)---can have utility arbitrarily greater than one.

Also, we will use $\lambda$ to denote the Lipschitz constant of the utility functions, $u_i$s. \\ 

\noindent 
{\bf Replica Economies.} Let $\Econ = \left( (u_i, \omega_i) \right)_{i \in [h]}$ be an exchange economy over the commodity space $\mathbb{R}^\ell_+$. The \emph{$n$-th replica} of $\Econ$, for $n\geq 1$, is the exchange economy $\Econ^n=\left( (u_{i, t},\omega_{i,t}) \right)_{i \in [n],  t \in [h]}$, with $nh$ consumers. In $\Econ^n$ the consumers are indexed by $(i,t)$, with index $i \in [n]$ and type $t \in [h]$, and they satisfy:
\[
u_{i,t} = u_t \text{ and } \omega_{i,t} = \omega_{t}.
\]

Replica economies constitute the basic model of a large economy with limited heterogeneity. See \cite[Chapter 18]{mas1995microeconomic} for a textbook treatment.\footnote{The idea of replicas was suggested by Edgeworth himself in his discussion of the competitive hypothesis in large economies. He used $X$ and $Y$ to refer to two different agents, and wrote:
``Let us now introduce a second X  and a second Y ; so that the field of competition consists of two Xs and two Ys. And for the sake of illustration (not of argument) let us suppose that the new X has the same requirements, the same nature as the old X; and similarly that the new Y is equal-natured with the old.''
\citep{edgeworth1881mathematical}.}

As mentioned previously, the utility normalization is considered only with respect to feasible allocations in the underlying economy $\Econ$.\footnote{This, in particular, ensures that this normalization is compatible with the strict monotonicity of the utilities.}  Indeed, the consumers' utilities scale up for feasible bundles in $\Econ^n$, since the amount of each good in the economy increases as a function of $n$. 

An allocation $\overline{x}=(\overline{x}_{i,t})$ of $\Econ^n$ is said to have the \emph{equal treatment property} iff all the consumers with the same type are allocated identical bundles, i.e., iff $\overline{x}_{i,t}= \overline{x}_{i',t}$ for every $t \in [h]$ and $i,i' \in [n]$. 

The following lemma is established in the work of Debreu and Scarf~\cite{debreu1963limit} and it asserts that allocations contained in the $h$-core necessarily satisfy the equal treatment property. 
Debreu and Scarf state the result as a consequence of the core, or, in the terminology of our paper,  the $(n\times h)$-core. In actuality, the stronger statement below holds true. We include a proof in Appendix~\ref{appendix:equal-treatment} for completeness, but the proof is really the same as that of Debreu and Scarf.

\begin{lemma}[Equal treatment property]
\label{lem:equaltreatment}
Let the utilities of each consumer be strictly monotonic, continuous, and strictly concave in an exchange economy $\Econ= \left( (u_i, \omega_i) \right)_{i \in [h]}$ and, hence, in the corresponding replica economy $\Econ^n$. Then, every $h$-core allocation of $\Econ^n$ satisfies the equal treatment property. 
\end{lemma}

By Lemma~\ref{lem:equaltreatment}, we can express the $h$-core allocations of $\Econ^n$ as vectors in $\mathbb{R}^{h\ell}_+$; with the convention that consumers that have the same type receive the same bundle. This succinct representation also applies to our main result, which considers $\kappa$-cores, with $\kappa \geq h$. Hence, for brevity, we represent $\kappa$-core allocations\footnote{Note that in $\Econ^n$ an arbitrary allocation (say, one that does not satisfy equal treatment) is, in fact, a vector in $\mathbb{R}^{nh\ell}_+$.} of the replica economy $\Econ^n$ as vectors in $\mathbb{R}^{h\ell}_+$. \\

The definition of a Walrasian equilibrium (both, exact and approximate) naturally extends to replica economies $\Econ^n$, wherein the price vector $p \in \mathbb{R}^\ell_+$ and the allocation $\overline{x} \in \mathbb{R}^{n h \ell}_+$. 

Furthermore, note that if $(p, \overline{x})$ is a (approximate) Walrasian equilibrium in $\Econ$, then $p$ and an $n$-times replicated version of $\overline{x}$ constitute a (approximate) Walrasian equilibrium in $\Econ^n$. Analogously, if $\overline{x}$ is a Walrasian allocation in $\Econ$, then it\footnote{Specifically, a replicated version of $\overline{x}$} is a Walrasian allocation in $\Econ^n$ as well, for all $n \geq 1$.

Throughout, we will address allocations that satisfy the equal treatment property. Hence, we will simply state that $(p, \overline{x}) \in \mathbb{R}^\ell_+ \times \mathbb{R}^{h \ell}_+$ is a (approximate) Walrasian equilibrium---or that $\overline{x} \in \mathbb{R}^{h\ell}_+$ is a (approximate) Walrasian  allocation---without distinguishing between $\Econ$ and $\Econ^n$ in this context.

%% file: quant-core-convergence.tex
\section{A Quantitative Core Convergence Theorem}
\label{section:quant-debreu-scarf}

\begin{quote}
    ``\ldots {\em the reason why the complex play of competition tends to a simple uniform result -- what is arbitrary and indeterminate in contract between individuals becoming extinct in the jostle of competition -- is to be sought in a principle which pervades all mathematics, the principle of limit, or law of great numbers as it might perhaps be called.}'' \\ (Francis Ysidro Edgeworth~\cite{edgeworth1884rationale}) \\
\end{quote} 

The following classic result of Debreu and Scarf~\cite{debreu1963limit} establishes that, as $n$ tends to infinity, each allocation in the core of $\Econ^n$ is a Walrasian allocation. 

\begin{theorem}[Debreu-Scarf Core Convergence Theorem]
Assume that in an exchange economy $\Econ$---with $h$ consumers and $\ell$ number of goods---the consumers' utilities are strictly monotonic, continuous, and strictly quasiconcave. Furthermore, suppose that an allocation $\overline{x} \in \mathbb{R}^{h\ell}_+$ belongs to the core of $\Econ^n$ for all $n \geq 1$. Then, $\overline{x}$ is a Walrasian allocation, i.e., there exists a nonzero price vector $p \in \mathbb{R}^\ell_+$ such that $(p, \overline{x})$ is a Walrasian equilibrium.  
\end{theorem}


By contrast, our result is nonasymptotic--it shows that as long as $n$ is quadratic in $h$, $\ell$, and $1/\varepsilon$, a core allocation is an $\varepsilon$-Walrasian allocation. In fact, our result holds for $\kappa$-cores, where $\kappa$ is $\mathcal{O}\left( \frac{h^2 \ell}{\varepsilon^2}\right)$.

\begin{theorem}[Main Result]
\label{thm:quant-debreu-scarf}
Assume that in an exchange economy $\Econ = \left( (u_i, \omega_i) \right)_{i \in [h]}$---with $h$ consumers and $\ell$ number of goods---the utilities, $u_i$s, are strictly monotonic, continuously differentiable, and $\alpha$-strongly concave. Furthermore, suppose that an allocation $\overline{x} \in \mathbb{R}^{h\ell}_+$ belongs to the $\kappa$-core of $\Econ^n$, for any
\begin{align*}
n \geq \kappa \geq \frac{16}{\alpha} \left( \frac{\lambda \ell h }{\varepsilon} + \frac{h^2}{\varepsilon^2} \right).
\end{align*}
Then, $\overline{x}$ is an $\varepsilon$-Walrasian allocation (i.e., there exists a price vector $p \in \Delta$ such that $(p, \overline{x})$ is an $\varepsilon$-Walrasian equilibrium). Here, $\varepsilon >0$ is the approximation parameter and $\lambda$ is the Lipschitz constant of the utilities. 
\end{theorem}

We will establish this theorem by showing that if an allocation $\overline{x}$ is in the $\kappa$-core, then there exists a price vector $p \in \Delta$ such that $(p, \overline{x})$ is a $\varepsilon$-Walrasian equilibrium. It is relevant to note that, instead of relying on limiting arguments, we use nonasymptotic results (with relevant approximation guarantees), such as the approximate version of Carath\'{e}odory's theorem (see~\cite{barman2015approximating} and references therein). Developing such a proof ensures that the underlying parameters (e.g., the core size $\kappa$) are only polynomially large (and not arbitrarily large). This approach not only provides us with a quantitive version of the core convergence theorem, but also an efficiently implementable characterization (Section~\ref{section:bounded-hull}). In Section~\ref{section:testing-walras} we will use the developed characterization to design an efficient algorithm that tests whether a given allocation is approximately Walrasian, or not.

We will begin by establishing (in Section~\ref{section:bounded-hull}) a useful geometric property that is satisfied by all allocations (Walrasian or otherwise). At a high level, this property shows that one can focus on a bounded subset of a specific convex hull, which in itself is unbounded. This bounding exercise essentially enables us to bypass asymptotic arguments, and prove the quantitative version of the core convergence theorem in Section~\ref{section:proof-quant-debreu-scarf}. 

The fact that we can work with a bounded set (and not an unbounded one) is also essential from an algorithmic standpoint. Specifically, it lets us apply the Ellipsoid method and develop (in Section~\ref{section:testing-walras}) an efficient, testing algorithm for Walrasian allocations.

\subsection{Bounded Hull}
\label{section:bounded-hull}
The result developed in this section applies to arbitrary allocations; allocations that might or might not be in the core. Consider an exchange economy $\Econ = \left( (u_i, \omega_i) \right)_{i \in [h]}$---with $h$ consumers and $\ell$ goods---in which the utilities, $u_i$s, are strictly monotonic, continuously differentiable, and $\alpha$-strongly concave. As before, $\lambda$ denotes the Lipschitz constant of the utility functions.  Let $\delta \coloneqq \varepsilon/h$.

Let $\overline{y} = (\overline{y}_i)_{i\in [h]}$ be an allocation in $\Econ$. With parameter $\eta \in (0,1)$, for each consumer $i \in [h]$ we define $V_i^\eta \coloneqq \left\{ y \in \mathbb{R}^\ell_+ \mid u_i(y) \geq u_i(\overline{y}_i) + \eta \right\}$ to denote the upper contour set with respect to the  allocated bundle $\overline{y}_i$ and margin $\eta$. Also, write $Q^\eta_i$ to denote the set of trades from the endowment $\omega_i$ that render consumer $i$ better off than the allocated bundle, $Q^\eta_i \coloneqq \left\{ z \in \mathbb{R}^\ell \mid u_i( z + \omega_i) \geq u_i(\overline{y}_i) + \eta  \right\}$. By definition, $z  \in Q^\eta_i$ iff $(z + \omega_i) \in V^\eta_i$.

For each consumer $i$, we also consider  $\widehat{Q}^\eta_i$, a bounded subset of $Q^\eta_i$; specifically, 
\begin{align*}
\widehat{Q}^\eta_i & \coloneqq Q^\eta_i  \ \cap \ \left\{ z \in \mathbb{R}^\ell \ : \ \| z \|  \leq  \sqrt{\frac{2 ( \lambda \ell \delta + 1)}{\alpha}} \right\},
\end{align*} 

The set $\widehat{Q}^\eta_i$ is obtained by intersecting $Q^\eta_i$ with the Euclidean ball of radius $\sqrt{\frac{2 ( \lambda \ell \delta + 1)}{\alpha}}$ and center $\zero$.

Since the utility $u_i$ is continuous and concave, the set $Q^\eta_i$ is closed and convex. Therefore, the subset $\widehat{Q}^\eta_i$ is compact (closed and bounded) and convex. The following lemma shows that, by bounding the sets in this manner, we do not not loose out on an important containment property.   


\begin{lemma}
\label{lemma:bounded-hull} Let  $\overline{y}$ be an allocation in an economy $\Econ$ with strictly monotonic, continuously differentiable, and strongly concave utilities. Suppose that the sets $Q^\eta_i$ and $\widehat{Q}^\eta_i$, for $i\in [h]$, are as defined above, with parameters $\da>0$ and $\eta\in [0,\bar\eta)$. Then,  
\begin{align*}
(-\delta) \one \in \cvh \left( \bigcup_{i=1}^h Q^\eta_i\right) & \ \ \ \textrm{  iff  } \ \ \ (-\delta) \one \in \cvh \left( \bigcup_{i=1}^h \widehat{Q}^\eta_i\right).
\end{align*}
\end{lemma}
\begin{proof}
The reverse direction of the claim is direct, since $\widehat{Q}^\eta_i \subset Q^\eta_i$ for all $i \in [h]$. 

For the forward direction, we have vectors $z_i \in Q^\eta_i$ and a convex combination $\lambda_i \geq 0$, for $i \in [h]$, such that $\sum_{i=1}^h \lambda_i = 1$ and
\begin{align}
\sum_{i=1}^h \lambda_i z_i = (-\delta) \one \label{eq:comb}
\end{align}

Let $R :=\max_i \{\| z_i \| :i\in [h]\}$. By definition, the $z_i$s are contained in the (closed) Euclidean ball $B(R)$ of radius $R$ and center $\zero$. 
Note that, for each $i \in [h]$, the intersection $Q^\eta_i \cap B(R)$ is a compact set. 

Let $Z$ denote the collection of all tuples \[(z'_1, z'_2, \ldots, z'_h) \in \left(Q^\eta_1 \cap B(R) \right) \times \left( Q^\eta_2 \cap B(R) \right) \times \ldots \times \left( Q^\eta_h \cap B(R)\right)\] for which there exists exists convex coefficients $\lambda'_i$s such that $\sum_i \lambda'_i z'_i \leq (-\delta) \one$, i.e., there exists a convex combination of $z'_i$s which is component-wise upper bounded by $(-\delta) \one$.\footnote{The definition of $Q_i^\eta$s provide a component-wise lower bound as well: $z_i \geq - \omega_i$, for each vector $z_i \in Q^\eta_i$.}

From~\eqref{eq:comb}, we know that $Z$ is nonempty. Given that the sets $(Q^\eta_i \cap B(R))$s are compact, one can show that $Z$ is compact as well. 
Hence, the problem of minimizing $\max\left\{ \|z_i' -  (-\delta) \one \| \ \mid (z'_i)_i \in Z \right\}$ admits an optimal solution, say $(z^*_t)_t$. Note that, by definition of $Z$, there exists convex coefficients $(\lambda^*_t)_{t \in H^*}$ that satisfy $\sum_{ t \in H^*} \lambda^*_t z^*_t \leq (-\delta) \one$;  
here, subset $H^* \subseteq [h]$ is selected to ensure that $\lambda^*_t > 0 $ for all $t \in H^*$. 
 
Next, we will prove that $\| z^*_t \| \leq \sqrt{\frac{2 ( \lambda \ell \delta + 1)}{\alpha}} $, for all $t \in H^*$. Subsequently, we will show that using $z^*_t$s we can obtain vectors $\widetilde{z}_t \in Q^\eta_t$ that satisfy the same norm bound ($\| \widetilde{z}_t \| \leq \sqrt{\frac{2 ( \lambda \ell \delta + 1)}{\alpha}} $) and whose convex combination is equal to $(-\delta) \one$. This norm bound implies that $\widetilde{z}_t \in \widehat{Q}^\eta_t$ and, hence, leads to the desired containment: $(-\delta) \one \in \cvh \left( \bigcup_{i=1}^h \widehat{Q}^\eta_i\right)$. 

Assume, by way of contradiction, that $\| z^*_i \| >  \sqrt{\frac{2 ( \lambda \ell \delta + 1)}{\alpha}}$, for some $i \in H^*$. For such an $i$, consider bundle $x^*_i \coloneqq z^*_i + \omega_i$. By definition, $z^*_i \in Q^\eta_i$ and, hence, $x^*_i \in V^\eta_i$ (i.e., $x^*_i$ belongs to the upper contour set). Given that $u_i$ is $\alpha$-strongly concave, we have\footnote{Note that, here the desired bound holds even if $x^*_i$ lies outside the range set $\mathcal{R}$ of strong concavity (see Section~\ref{section:notation}). Specifically, we can consider a convex combination of $\omega_i$ and $x^*_i$, say vector $\widetilde{x}_i$, that is at a distance $r$ (as specified in Section~\ref{section:notation}) away from $\omega_i$. The assumption on $\mathcal{R}$ ensures that $\widetilde{x}_i \in \mathcal{R}$ and $\frac{\alpha r^2}{2} \geq \lambda \delta \ell + 1$. Applying strong concavity with $\widetilde{x}_i$ we get $u_i(\omega_i) \leq u_i(\widetilde{x}_i) + \nabla u_i(\widetilde{x}_i)^T( \omega_i - \widetilde{x}_i) - \frac{\alpha}{2} \| \omega_i - \widetilde{x}_i \|^2 \leq  u_i(\widetilde{x}_i) + \nabla u_i(\widetilde{x}_i)^T( \omega_i - \widetilde{x}_i) - (\lambda \ell \delta + 1)$. Since $u_i(x^*_i) + \nabla u_i(x^*_i)^T( \omega_i - x^*_i) \geq u_i(\widetilde{x}_i) + \nabla u_i(\widetilde{x}_i)^T( \omega_i - \widetilde{x}_i)$, inequality (\ref{ineq:interim}) follows.}
  
\begin{align}
u_i(\omega_i) & \leq u_i(x^*_i) + \nabla u_i(x^*_i)^T( \omega_i - x^*_i) - \frac{\alpha}{2} \| \omega_i - x^*_i \|^2 \label{ineq:defn-strong-concave}
\end{align}
Since $\| \omega_i - x^*_i \| = \| z^*_i \| > \sqrt{\frac{2 ( \lambda \ell \delta + 1)}{\alpha}}$, inequality (\ref{ineq:defn-strong-concave}) reduces to 
\begin{align}
\nabla u_i(x^*_i)^T x^*_i < \nabla u_i(x^*_i)^T \omega_i + \left( u_i(x^*_i) - u_i(\omega_i) \right) - \frac{\alpha}{2} \frac{2 ( \lambda \ell \delta + 1)}{\alpha} \label{ineq:interim}
\end{align}

Now observe that $x^*_i$ must satisfy $u_i(x^*_i)=u_i(\bar y_i)+\eta < 1$: if this is not the case (i.e., we have $u_i(x^*_i) > u_i(\bar y_i)+\eta$), then by reducing a positive component\footnote{If all the components of $z^*_i$ are negative, then $x^*_i \leq \omega_i$ and we get the desired bound $u_i(x^*_i) \leq u_i(\omega_i) \leq 1 - \bar{\eta}$. Here, the last inequality follows from the fact that the utility of any feasible allocation (and, hence, for $\omega_i$) is normalized to be at most $1 - \bar{\eta}$} of $z^*_i = x^*_i - \omega_i$ we can ensure that $z^*_i$ moves closer to $(-\delta) \one$ and at the same time $z^*_i$ continues to be in $Q^\eta_i$ (i.e., the following inequality continues to hold $u_i(x^*_i) \geq u_i(\bar y_i)+\eta$). Also, note that such a reduction maintains the containment of $z^*_i$s in $Z$; specifically, the following inequality continues to hold $\sum_{ t } \lambda^*_t z^*_t \leq (-\delta) \one$. A repeated application of this argument gives us $x^*_i = z^*_i + \omega_i$ with the property that $u_i(x^*_i)=u_i(\bar y_i)+\eta < 1$; here, we have the last inequality since $\eta<\bar\eta$.\footnote{Recall that the utility of any feasible allocation (and, hence, for $\overline{y}_i$) is normalized to be at most $1 - \bar{\eta}$.}


The utility function $u_i$ is $\lambda$-Lipschitz, hence, its gradient at any point $x^*_i \in \mathbb{R}^\ell$ satisfies $\| \nabla u_i(x^*_i) \|_\infty \leq \lambda$. That is, $\| \nabla u_i(x^*_i) \|_1 \leq \lambda \ell$. Therefore, 
\begin{align}
\nabla u_i(x^*_i)^T x^*_i & < \nabla u_i(x^*_i)^T \left( \omega_i +  (-\delta) \one \right) \label{ineq:gradient-drop}
\end{align}

Now, we can apply Proposition~\ref{prop:gradient} (stated and proved below) with $x = x^*_i$, $w = (\omega_i +  (-\delta) \one)$, and inequality (\ref{ineq:gradient-drop}), to establish that there exists a positive $\mu \in (0,1]$ such that $u_i (x^*_i) \leq u_i\left( (1-\mu) x^*_i + \mu \left(  \omega_i +  (-\delta) \one \right)  \right)$.

Rewriting, we get $u_i\left( (1- \mu) ( x^*_i - \omega_i) + \mu (-\delta) \one + \omega_i \right) \geq u_i(x^*_i) $. Therefore, \[\widehat{z}_i \coloneqq (1- \mu) ( x^*_i - \omega_i) + \mu (-\delta) \one = (1- \mu)  z^*_i + \mu (-\delta) \one \in Q^\eta_i,\] as $x^*_i \in V^\eta_i$ and $z^*_i \in Q^\eta_i$. 


The vector $\widehat{z}_i$ is itself a convex combination of $z^*_i$ and $(-\delta) \one$. This ensures that there exists a convex combination, $(\widehat{\lambda}_t)_{t \in H^*}$, such that 
\begin{align*}
\widehat{\lambda}_i \widehat{z}_i + \sum_{t \in H^* \setminus \{ i \}} \widehat{\lambda}_t z^*_t  \leq (-\delta) \one.
\end{align*}

Note that $\|\widehat{z}_i - (-\delta) \one \|  = (1 - \mu) \| z^*_i - (-\delta) \one \| < \| z^*_i - (-\delta) \one\|$. This argument can be repeated for all other $j$s with $\| z^*_{j} - (-\delta) \one \| = \max\{ \| z^*_t - (-\delta) \one \| \ \mid \ t \in H^* \}$, contradicting the definition (optimality) of $z^*_t$s. 

Therefore, the desired upper bound on the norm holds: $\| z^*_t \| \leq \sqrt{\frac{2 ( \lambda \ell \delta + 1)}{\alpha}} $, for all $t \in H^*$. 

To complete the proof we will show that $z^*_t$s can be transformed into vectors $\widetilde{z}_t \in Q^\eta_t$ that satisfy the same norm bound and whose convex combination is equal to $(-\delta) \one$. Write $\varphi \coloneqq \sum_t \lambda^*_t z^*_t $ and note that $\varphi \leq (-\delta) \one$. If component $a \in [\ell]$ of $\varphi$ is strictly less than $-\delta$, then there exists a $z^*_i$ such that its $a$th component is less than $-\delta$ as well: $z^*_{i, a} < -\delta$. We can increase $z^*_{i,a}$ till either it becomes equal to zero, or the $a$th component of $\varphi$ reaches $-\delta$.\footnote{Here, the convex coefficients, $\lambda^*_t$s, remain unchanged.} Note that in this transformation while the $a$th component of $z^*_i$ increases in value, it decreases in magnitude. Hence, the utility $u_i(z^*_i + \omega_i)$ increases and the norm of $z^*_i$ decreases. Repeatedly applying this procedure gives us vectors $\widetilde{z}_t \in Q^\eta_t$ such that $\| \widetilde{z}_t \| \leq \Lambda$ and $\sum_t \lambda^*_t \widetilde{z}_t  = (-\delta) \one$. 
 
Overall, this implies that $(-\delta) \one \in \cvh \left( \bigcup_{i=1}^h \widetilde{Q}_i\right)$ and the stated claim follows. 
\end{proof}

The following observation was used in the proof of Lemma~\ref{lemma:bounded-hull}. 

\begin{proposition}
\label{prop:gradient}
For a continuously differentiable and concave function $u: \mathbb{R}^\ell \mapsto \mathbb{R}$, let vector $x, w \in \mathbb{R}^\ell$ satisfy the inequality $\nabla u(x)^T x < \nabla u(x)^T w$. Then, there exists a positive $\mu \in (0,1]$ such that $u(x) \leq u\left( (1-\mu) x + \mu w\right)$.  
\end{proposition}
\begin{proof}
Let $\zeta \coloneqq \nabla u(x)^T (w - x) >0$. Also, for parameter $\mu \in (0,1]$, denote by $x_\mu$ the convex combination $(1-\mu) x + \mu w$. As $\mu$ tends to zero, $x_\mu$ tends to $x$. 

Given that $u$ is concave, 
\begin{align}
u(x) & \leq u(x_\mu) + \nabla u(x_\mu)^T ( x - x_\mu) \nonumber \\ 
& = u(x_\mu) - \mu \nabla u(x_\mu)^T ( w - x) \label{ineq:grad}
\end{align}

The Cauchy-Schwarz inequality gives us $| \langle \nabla u(x_\mu) - \nabla u(x),  (w - x) \rangle | \leq \| \nabla u(x_\mu) - \nabla u(x)  \|  \cdot \| w -x\| $. Since $u$ is continuously differentiable, there exists a small enough, but positive, $\mu$ such that the right-hand-side of the previous inequality is strictly less than $\zeta$. For such a $\mu >0$ we have $\nabla u(x_\mu)^T ( w - x) > 0$. Therefore, using inequality (\ref{ineq:grad}), we get that there exists a $\mu \in (0,1]$ for which $u(x) \leq u\left( (1-\mu) x + \mu w\right)$.

\end{proof}

\subsection{Lemma for $\kappa$-Core Allocations}
Recall the notation from Theorem~\ref{thm:quant-debreu-scarf} and write $\delta \coloneqq \frac{\varepsilon}{h}$ along with $\gamma \coloneqq \sqrt{\frac{2 ( \lambda \ell \delta + 1)}{\alpha}}$. Note that, in Theorem~\ref{thm:quant-debreu-scarf}, the lower bound on $\kappa$ is equal to $\frac{8 \gamma^2}{\delta^2}$. 

In this section we apply Lemma~\ref{lemma:bounded-hull} to establish an important separation property satisfied by allocations in the $\kappa$-core of an economy. In contrast to the result developed in the previous section, the next lemma (Lemma~\ref{lemma:convex-hull}) specifically addresses $\kappa$-core allocations.  

In particular, given a $\kappa$-core allocation $\overline{x} =(\overline{x}_i)_{i \in [h]} \in \mathbb{R}^{h\ell}_+$, let $\bar{\eta} > 0 $ be such that $u_i(\overline{x}_i) \in [0, 1 - \bar{\eta}]$ for all $i \in [h]$--the normalization of the utilities (to lie in $[0,1)$) ensures that such a $\bar{\eta}$ exists. Consider an arbitrarily small, but positive, parameter $\eta \in (0, \bar{\eta})$. 

For each consumer $i \in [h]$, write $U_i^\eta \coloneqq \left\{ x \in \mathbb{R}^\ell_+ \mid u_i(x) \geq u_i(\overline{x}_i) + \eta \right\}$ to denote the upper contour set with respect to the  allocated bundle $\overline{x}_i$ and margin $\eta$. Also, write $P^\eta_i$ to denote the set of trades from the endowment $\omega_i$ that render consumer $i$ better off than the allocated bundle, $P^\eta_i \coloneqq \left\{ z \in \mathbb{R}^\ell \mid u_i( z + \omega_i) \geq u_i(\overline{x}_i) + \eta  \right\}$. By definition, $z  \in P^\eta_i$ iff $(z + \omega_i) \in U^\eta_i$. 


The next lemma provides an important characterization in terms of the convex hull of the $P^\eta_i$s. This lemma shows that under a $\kappa$-core allocation beneficial trades (i.e., vectors in $P^\eta_i$s) cannot be combined (as a convex combination) across consumers to obtain $(-\delta) \one$; recall that $\delta \coloneqq \varepsilon/h$. At a high level, this corresponds to the fact that, under a $\kappa$-core allocation, mutually beneficial redistributions are not possible across consumers. 

\begin{lemma}
\label{lemma:convex-hull}
For any allocation $\overline{x}=(\overline{x}_i)_{i \in [h]}$ that belongs to the $\kappa$-core of an economy $\Econ^n$,  we have
\begin{align*}
\left(-\delta \right) \one \notin \cvh \left( \bigcup_{i=1}^h P^\eta_i\right).
\end{align*}
\end{lemma}
\begin{proof}
Write $\nu \coloneqq (-\delta) \one$. Suppose, towards a contradiction, that $\nu \in \cvh \left( \bigcup_{i=1}^h P^\eta_i\right)$. Then, applying Lemma~\ref{lemma:bounded-hull}, we get that there exist vectors $z^*_t \in P^\eta_t$ such that their norm is bounded, $\| z^*_t \| \leq \sqrt{\frac{2 ( \lambda \ell \delta + 1)}{\alpha}}$, and their convex hull contains $\nu$: 
\begin{align*}
\nu \in \cvh \ (z^*_t)_{t \in H^*} 
\end{align*}
Here, $H^*$ is a subset of $[h]$ with the property that all $t$ in $H^*$ have strictly positive weight in the convex combination that yields $\nu$. 

Therefore, via the approximate version of Carath\'{e}odory's theorem (see, in particular, Theorem 2 in~\cite{barman2015approximating}), we get that, for any integer $k \geq \frac{8 \gamma^2}{\delta^2}$, there exists a vector $\nu' \in \cvh \ (z^*_t)_{t \in H^*}$ which is $\delta$ close to $\nu$ (i.e., $\| \nu' - \nu \|  < \delta$) and it satisfies   
\begin{align*}
\nu'  = \sum_{t \in H^*} \frac{\beta_t}{k} z^*_t,
\end{align*}
with integers $\beta_t \in \mathbb{Z}_+$ summing up to $k$: $\sum_{t \in H^*} \beta_t = k$.

Choose an integer $k \geq \frac{8 \gamma^2}{\delta^2}$ with $k\leq \kappa$. Then, $\nu'$ is $\delta$-close to $\nu$ and it is a $k$-uniform convex combination of the $z^*_t$s. Given that the Euclidean distance between $\nu'$ and $\nu= (-\delta) \one $ is strictly smaller than $\delta$,  the vector $\nu'$ is componentwise strictly negative: $\nu' \ll \zero$. 

Now, consider a coalition $S$ in $\Econ^n$ defined by including $\beta_t \in \mathbb{Z}_+$ copies of consumers of type $t \in H^*$. Note that $|S| = k \leq \kappa$. 

For each member of $S$ of type $t$, let consumption bundle $y^*_{i, t} \coloneqq z^*_t + \omega_t$. Since $z^*_t \in P^\eta_t$, we have 
\begin{align*}
u_t( y^*_{i,t}) \geq u_t(\overline{x}_t) + \eta \qquad \text{ for all } (i, t) \in S.
\end{align*}
Also, the fact that $\nu'$ is a $k$-uniform convex combination of the $z_t^*$s gives us 
\begin{align*}
\frac{1}{k} \left( \sum_{(i,t) \in S} y^*_{i, t} - \sum_{(i,t) \in S} \omega_t \right) = \nu' \ll \zero
\end{align*}

Therefore, $\sum_{(i,t) \in S} y^*_{i, t} \ll  \sum_{(i,t) \in S} \omega_t = \sum_{(i,t) \in S} \omega_{i, t}$. As a consequence, the coalition $S$, which has cardinality $k \leq \kappa$, blocks the given allocation $\overline{x} = (\overline{x}_i)_{i \in [h]}$, in the replica economy $\Econ^n$. This contradicts that $\overline{x}$ belongs to the $\kappa$-core of $\Econ^n$.  
\end{proof}

\subsection{Proof of Theorem~\ref{thm:quant-debreu-scarf}}
\label{section:proof-quant-debreu-scarf}
For the given $\kappa$-core allocation $\overline{x} =(\overline{x}_i)_{i \in [h]} \in \mathbb{R}^{h\ell}_+$ of the economy $\Econ^n$, we have $U^\eta_i \coloneqq \{x \in \mathbb{R}^\ell_+ \mid u_i(x) \geq u_i(\overline{x}_i) + \eta \} $ and $P^\eta_i \coloneqq \{ z \in \mathbb{R}^\ell \mid u_i(z + \omega_i) \geq u_i(\overline{x}_i) \}$, for each $i \in [h]$.

Lemma \ref{lemma:convex-hull} guarantees that $\left(-\delta \right) \one \notin \cvh \left( \bigcup_{i=1}^h P^\eta_i\right)$. Hence, there exists a hyperplane between $ \left(-\delta \right) \one$ and the convex hull; an implication of the hyperplane separation theorem.

In particular, let $p \in \mathbb{R}^\ell$ specify such a separating hyperplane: $p \cdot \left( (-\delta) \one \right) \leq p \cdot  \cvh \left( \bigcup_{i=1}^h P^\eta_i\right)$. The upward closure of $P^\eta_i$s (a consequence of strict monotonicity of the consumers' utility functions), and $p\neq 0$, ensures that $ p > \zero$. Hence, by scaling, we can assume that $p \in \Delta$. 

We will show that $p$---as a price vector---certifies that $\overline{x}$ is an $\varepsilon$-Walrasian allocation, i.e., $(p, \overline{x})$ is an $\varepsilon$-Walrasian equilibrium. 

Recall that, in an exchange economy $\Econ= \left( (u_i, \omega_i) \right)_{i \in [h]}$, a pair $(q, \overline{y}) \in \Delta \times \mathbb{R}^{h \ell}_+$ (with price vector $q$ and allocation $\overline{y}$) is deemed to be an $\varepsilon$-Walrasian equilibrium iff the following two conditions hold for every consumer $i$: (i) $|q^T \overline{y}_i - q^T \omega_i| \leq \varepsilon$ and (ii) for any bundle $x$, with $u_i(x) > u_i (\overline{y}_i)$,  we have $q^T x > q^T \omega_i - \varepsilon/h$. To complete the proof, we will show that $(p, \overline{x})$ satisfies these conditions.

Note that, for each $i \in [h]$ and any $z = (x - \omega_i) \in P^\eta_i$, the separation by $p \in \Delta$ implies  
\begin{align*}
p^T (x - \omega_i) = p^T z \geq p^T \left( (-\delta) \one \right) = -\delta.
\end{align*}
Hence, for any bundle $x \in U^\eta_i$ (i.e., for any bundle $x$ that satisfies $u_i(x) \geq u_i(\overline{x}_i) + \eta$) the expenditure is at least the income (minus $\delta$): $p^T x \geq p^T \omega_i - \delta$; recall that $\delta = \varepsilon/h$. 

Here, the analysis holds for any $\eta>0$, however small. That is, for a bundle $x \in \mathbb{R}^\ell_+$, with the property that, $u_i(x) > u_i(\overline{x}_i)$, we have 
\begin{align}
p^T x \geq p^T \omega_i - \delta \label{ineq:Walras-two}
\end{align}
Therefore, $(p, \overline{x})$ satisfies the second condition in the definition of an $\varepsilon$-Walrasian equilibrium.  

Finally, we will show that even under the allocated bundle $\overline{x}_i$, the expenditure is close to the income. Using the continuity of the utilities and a small enough $\eta$, we can apply inequality (\ref{ineq:Walras-two}) to obtain $p^T \overline{x}_j \geq p^T \omega_j - \delta$ for all consumers $j \in [h]$. 

Allocation $\overline{x}=(\overline{x}_j)_{j \in [h]}$ satisfies the equal treatment property (Lemma~\ref{lem:equaltreatment}), hence $\sum_{j \in [h]} \overline{x}_j = \sum_{j \in [h]} \omega_j$. Therefore, for each consumer $i$, we have $(\overline{x}_i - \omega_i) = \sum_{j \in [h] \setminus \{i\}} (\omega_j - \overline{x}_j)$. Taking inner product with $p \in \Delta$, we get that $(p, \overline{x})$ satisfies the first condition in the definition of an $\varepsilon$-Walrasian equilibrium as well:
\begin{align*}
p^T ( \overline{x}_i - \omega_i) = \sum_{j \in [h] \setminus \{i\}} p^T (\omega_j - \overline{x}_j) \leq (h - 1) \delta < \varepsilon.
\end{align*}
 
Overall, we get that $(p, \overline{x})$ is an $\varepsilon$-Walrasian equilibrium and this completes the proof.

%% file: testing-algorithm.tex
\section{Testing Algorithm for Walrasian Allocations}
\label{section:testing-walras}

This section develops a polynomial-time algorithm that efficiently determines whether a given allocation is an $\varepsilon$-Walrasian allocation, or not. Specifically, given an exchange economy\footnote{Our algorithm only requires oracle access to the underlying utilities $u_i$s and their gradients. In particular, our algorithmic result will hold even in the absence of an explicit (say, a closed form) description of the utility functions.} $\Econ = \left( (u_i, \omega_i) \right)_{i \in [h]}$---with $h$ consumers and $\ell$ goods---along with an allocation $\overline{y}=(\overline{y}_i)_{i \in [h]} \in \mathbb{R}^{h \ell}$, the developed algorithm efficiently finds a price vector $p \in \Delta$ (if one exists) such that $(p, \overline{y})$ is an $\varepsilon$-Walrasian equilibrium in $\Econ$. If no such price vector exists (i.e., $\overline{y}$ is not an $\varepsilon$-Walrasian allocation), then the algorithm correctly reports as such. 

The developed algorithm runs in time that is polynomial in the number of consumers $h$. The algorithm applies, in particular, to completely heterogeneous economies, in which all the consumers can be of different type. In other words, our algorithmic results are not confined to the replica-economy framework.     
 

The testing algorithm builds upon Lemma~\ref{lemma:bounded-hull}. As before, for any given allocation $\overline{y}=(\overline{y}_i)_{i \in [h]}$ (which might or might not be approximately Walrasian) and each consumer $i \in [h]$, we define the set $Q_i \coloneqq \{ z \in \mathbb{R}^\ell \mid u_i(z + \omega_i) \geq u_i(\overline{y_i}) \}$.\footnote{That is, in terms of the notation developed in Section~\ref{section:bounded-hull}, we are considering $Q^\eta_i$ with $\eta = 0$.} We will prove (in Lemma~\ref{lemma:non-containment} below) that the (non) containment of $(-\delta) \one$ in the convex hull of the $Q_i$s characterizes Walrasian equilibria; recall that $\delta \coloneqq \varepsilon/h$. 

This geometric characterization is interesting in its own right. However, given that the sets $Q_i$s are unbounded, this characterization (in terms of $Q_i$s), does not, in and of itself, translate into an efficient testing algorithm. Specifically, in order to apply the Ellipsoid method and test whether a vector is contained in a specific convex set, one requires the set to be bounded. The quantitive treatment developed in this paper---in particular, Lemma~\ref{lemma:bounded-hull}---enables us to bypass this issue. Specifically, we will show that it suffices to work with the a bounded subset of $Q_i$. Towards that end, we define
\begin{align}
\widehat{Q}_i & \coloneqq Q_i  \ \cap \ \left\{ z \in \mathbb{R}^\ell \ : \ \| z \| \leq  \sqrt{\frac{2 ( \lambda \ell \delta + 1)}{\alpha}}  \right\} \label{eq:definition}
\end{align} 
Here, $\alpha$ is the strong-concavity parameter of the utilities and $\lambda$ is their Lipschitz constant. 

As observed before, the set $Q_i$ is closed and convex. Furthermore, the subset $\widehat{Q}_i$ is compact (closed and bounded) and convex.\footnote{We can enlarge $\widehat{Q}_i$ to ensure that it always has a nonempty interior. For instance, we can set the the radius of the intersecting Euclidean ball to be, say, $\max\left\{\sqrt{\frac{2 ( \lambda \ell \delta + 1)}{\alpha}}, \ 2 \| \overline{y}_i\| \right\}$; note that $\overline{y}_i \in Q_i$ and, hence, with this redefinition we have $\overline{y}_i \in \widehat{Q}_i$.} Also, note that any vector $z$ in $Q_i$ of norm less than $\sqrt{\frac{2 ( \lambda \ell \delta + 1)}{\alpha}}$, belongs to $\widehat{Q}_i$ as well.

\begin{lemma}
\label{lemma:non-containment}
An allocation $\overline{y}$ is an $\varepsilon$-Walrasian allocation in an economy $\Econ$ iff 
\begin{align*}
\left(-\delta \right) \one \notin \cvh \left( \bigcup_{i=1}^h \widehat{Q}_i \right).
\end{align*}
Here, for each consumer $i \in [h]$, the set $\widehat{Q}_i$ is as defined above.
\end{lemma}

\begin{proof}
First, we will consider the case wherein $(-\delta) \one \in \cvh \left( \bigcup_{i=1}^h \widehat{Q}_i \right)$ and show that this containment implies that $\overline{y}$ is not an $\varepsilon$-Walrasian allocation. Then, we will address the complementary case and prove that if  $(-\delta) \one \notin \cvh \left( \bigcup_{i=1}^h \widehat{Q}_i \right)$, then $\overline{y}$ is indeed an $\varepsilon$-Walrasian allocation.

The containment $(-\delta) \one \in \cvh \left( \bigcup_{i=1}^h \widehat{Q}_i \right)$ implies that there exists vectors $z_i \in \widehat{Q}_i \subset Q_i$, with $i \in [h]$, along with convex coefficients $\lambda_i$s such that $\sum_i \lambda_i z_i = (-\delta) \one$. 
Here, $\lambda_i \geq 0$, for all $i$, and $\sum_i \lambda_i = 1$. 

For each $i$, write $x_i = z_i + \omega_i$. Hence, we have 
\begin{align}
\sum_i \lambda_i (x_i - \omega_i) = (-\delta) \one \label{eq:for-inner-prod}
\end{align}

Furthermore, given that $z_i \in Q_i$, the definition of $Q_i$ implies $u_i(x_i) \geq u_i(\overline{y}_i)$, for each $i$.  

Say, towards a contradiction, that $\overline{y}$ is an $\varepsilon$-Walrasian allocation. That is, there exits a price vector $p \in \Delta$ such that $(p, \overline{y})$ is a $\varepsilon$-Walrasian equilibrium. Taking inner product with $p$ on both sides of the equation (\ref{eq:for-inner-prod}), we obtain $\sum_i \lambda_i \ p^T (x_i - \omega_i) = - \delta$. Since $\lambda_i$s are convex coefficients, there exists an index $j$ such that $p^T x_j \leq p^T \omega_j - \delta$.\footnote{With a slightly smaller value of $\delta$ we can directly obtain a strict inequality $p^T x_j  < p^T \omega_j - \varepsilon/h$.} Since $u_j(x_j) \geq u_j(\overline{y}_j)$, the inequalities contradict the fact (in particular, contradict the second condition in the definition of an approximate Walrasian equilibrium) that $(p, \overline{y})$ is an $\varepsilon$-Walrasian equilibrium. 

To complete the proof we now consider the complementary case: $(-\delta) \one \notin \cvh \left( \bigcup_{i=1}^h \widehat{Q}_i \right)$. Using Lemma~\ref{lemma:bounded-hull} (in contrapositive form, with $\eta=0$) we get $(-\delta) \one \notin \cvh \left( \bigcup_{i=1}^h {Q}_i \right)$. 
 
Consider a hyperplane between $(-\delta) \one$ and this convex hull: $p \cdot \left( (-\delta) \one \right) \leq p \cdot  \cvh \left( \bigcup_{i=1}^h Q_i\right)$. The hyperplane separation theorem guarantees the existence of such a $p \in \mathbb{R}^\ell$. Since $Q_i$ is upward closed, $p$ is componentwise nonnegative. Therefore, and given that $p\neq 0$, by scaling we can assume that $p \in \Delta$.  Next, we will show that $(p, \overline{y})$ is an $\varepsilon$-Walrasian equilibrium.

For each consumer $j \in [h]$ and any $z \in Q_j$, consider the vector $ y = z + \omega_j$. The definition of $Q_j$ implies $u_j(y) \geq u_j(\overline{y}_j)$. 
The separating property of hyperplane $p$ gives us $p^T(y - \omega_j) = p^T z  \geq p^T(-\delta)\one = -\delta$. Therefore, for each $j \in [h]$, the second condition in the definition of an approximate Walrasian equilibrium is satisfied: for any $y$ with the property that $u_j(y) > u_j(\overline{y}_j)$ we have 
\begin{align}
p^T y \geq p^T \omega_j - \delta \label{ineq:second-condition}
\end{align}

Inequality (\ref{ineq:second-condition}) holds, in particular, for $\overline{y}_j$s: $p^T \overline{y}_j \geq p^T \omega_j - \delta$. In addition, given that $\overline{y}$ is an allocation in the economy $\Econ$  we have that, under $\overline{y}$, the supply is equal to the demand: $\sum_{j =1}^h \overline{y}_j = \sum_{j=1}^h \omega_j$. 

Hence, for any fixed $i$, $\overline{y}_i - \omega_i = \sum_{j \in [h] \setminus \{ i \}} (\omega_j - \overline{y}_j )$. Multiplying both sides of this equality with $p^T$ we establish the first condition that defines an $\varepsilon$-Walrasian equilibrium:

\begin{align*}
p^T( \overline{y}_i - \omega_i ) & = \sum_{j \in [h] \setminus \{ i \}} p^T(\omega_j - \overline{y}_j) \\
& \leq \sum_{j \in [h] \setminus \{ i \}}  \delta \tag{using (\ref{ineq:second-condition}) with $y = \overline{y}_j$} \\
&  = (h-1) \delta < \varepsilon
\end{align*} 

That is, $|p^T \overline{y}_i - p^T \omega_i | \leq \varepsilon$. Hence, $(p, \overline{y})$ is an $\varepsilon$-Walrasian equilibrium and the stated claim follows.  
\end{proof}

In light of Lemma~\ref{lemma:non-containment}, testing whether an allocation $\overline{y}$ is approximately Walrasian, or not, reduces to determining whether the vector $(-\delta) \one$ is contained in the convex hull of the $\widehat{Q}_i$s. Below, in Theorem~\ref{theorem:testing-algorithm}, we develop an efficient separation oracle for the convex hull of the $\widehat{Q}_i$s--the testing algorithm is obtained by simply applying this separation oracle onto $(-\delta) \one$. 

For designing the efficient oracle, we use the equivalence of optimization and separation~\cite{grotschel2012geometric}. We show that (linear) optimization problems can be solved in polynomial time over the convex hull of the $\widehat{Q}_i$s. Therefore, we obtain the desired separation oracle. Note that this is a somewhat atypical application of the optimization-separation equivalence--we start with an optimization algorithm to obtain a separating one.   

The running time of our algorithm is polynomial in the input size; in particular, the running time is polynomial in the bit complexity of the underlying parameters (including $\varepsilon$). Furthermore, the algorithm only requires oracle access to the utilities and their gradients. 
\begin{theorem}[Testing Algorithm]
\label{theorem:testing-algorithm}
Let $\Econ$ be an exchange economy with monotonic, continuously differentiable, and strongly concave utilities. Then, there exists a polynomial-time algorithm that, given an allocation $\overline{y}$ in $\Econ$, determines whether $\overline{y}$ is an $\varepsilon$-Walrasian allocation.  
\end{theorem}

\begin{proof}
As a direct consequence of Lemma~\ref{lemma:non-containment} we have: testing for approximately Walrasian allocation corresponds to determining whether the vector $(-\delta) \one$ is contained in the convex hull of the $\widehat{Q}_i$s; see equation (\ref{eq:definition}) for the definition of these sets. 

Write $\widehat{\mathcal{Q}} \coloneqq \cvh \left( \bigcup_{i=1}^h \widehat{Q}_i \right)$. We will develop an efficient algorithm, \textsc{Alg}, for solving linear optimization problems over $\widehat{\mathcal{Q}}$, i.e., for solving problems of the form 
\begin{align} 
\max \ c^Tz & \ \ \  \text{ subject to } \ z \in \widehat{\mathcal{Q}} \label{eq:LP}
\end{align}
Here, $c \in \mathbb{R}^\ell$ is an input vector. 

The equivalence of optimization and separation (see, e.g.,~\cite{grotschel2012geometric}) implies that \textsc{ALG} can be used to design a polynomial-time algorithm \textsc{Sep} that provides a separation oracle for $\widehat{\mathcal{Q}}$. That is, using \textsc{Sep} we can perform the desired test of determining whether $(-\delta) \one \in \widehat{\mathcal{Q}}$, or not. 

In order to apply the optimization-separation equivalence we need to ensure that $\widehat{\mathcal{Q}}$ is compact, convex, and has a nonempty interior. These properties are satisfied by $\widehat{Q}_i$s individually, hence they hold for $\widehat{\mathcal{Q}}$ as well. Therefore, we can evoke the equivalence (via an application of the Ellipsoid method over the \emph{polar} of $\widehat{\mathcal{Q}}$) and obtain the algorithm \textsc{Sep}. 

To develop the algorithm, \textsc{Alg}, that efficiently solves linear optimization problems of the form (\ref{eq:LP}), we note that the feasible set $\widehat{\mathcal{Q}}$ is a convex hull of the $\widehat{Q}_i$s. Hence, for any $c \in \mathbb{R}^\ell$, an optimal solution of (\ref{eq:LP}) can be obtained by solving 
\begin{align}
\max_{i \in [h]} \left( \max \ c^Tz_i  \ \ \  \text{ subject to } \ z_i \in \widehat{Q}_i \right) \label{eq:LP1}
\end{align}
Here, for each $i$, the decision variable $z_i \in \mathbb{R}^\ell$ lies in the set $\widehat{Q}_i$. Below we will provide, for each $i$, a polynomial-time algorithm, $\textsc{Alg}_i$, that solves the linear optimization problem over $\widehat{Q}_i$, i.e., $\textsc{Alg}_i$ efficiently solves $\max \ c^Tz_i  \ \  \text{ subject to } \ z_i \in \widehat{Q}_i$. Hence, $\textsc{Alg}$ can be obtained by directly taking a maximum over the (optimal) solutions obtained by the $\textsc{Alg}_i$s. 

We will now complete the chain of arguments mentioned above by designing the optimization algorithm $\textsc{Alg}_i$. This algorithm is itself based on the Ellipsoid method. As detailed below, the gradients of the utility function $u_i$ (at different points) can be used to separate  $\widehat{Q}_i$ from vectors that are not contained in it. Hence, with this separation technique in hand, we can apply the Ellipsoid method over $\widehat{Q}_i$ to obtain $\textsc{Alg}_i$.\footnote{Recall that, $\widehat{Q}_i$s are compact, convex, and have a nonempty interior. Hence, the Ellipsoid method is applicable over these sets.}  	

Given a query vector $q \in \mathbb{R}^\ell$, it is easy to test if $q \in \widehat{Q}_i  \coloneqq Q_i \cap \ \left\{ z \in \mathbb{R}^\ell \ : \ \| z \| \leq  \sqrt{\frac{2 ( \lambda \ell \delta + 1)}{\alpha}}  \right\} $. We directly verify (i) $u_i(q + \omega_i) \geq u_i(\overline{y}_i)$  (to ensure that $q \in Q_i$) and (ii) $\| q \| \leq \sqrt{\frac{2 ( \lambda \ell \delta + 1)}{\alpha}}$.  

Consider the case in which $q \notin \widehat{Q}_i $. To run the Ellipsoid method (that underlies $\textsc{Alg}_i$), we need a separating hyperplane for such a $q \in \mathbb{R}^\ell$. There are two complementary (though, nonexclusive) cases either (i) $q \notin Q_i$ (i.e., $u_i(q + \omega_i) < u_i(\overline{y}_i)$) and (ii) $\| q \| > \sqrt{\frac{2 ( \lambda \ell \delta + 1)}{\alpha}}$.

In case (i), the gradient at $q + \omega_i$ (i.e., $\nabla u_i ( q + \omega_i) \in \mathbb{R}^\ell_+$) provides the separating hyperplane: utility $u_i$ is concave, hence $u_i(z + \omega_i) \leq u_i(q + \omega_i) + \nabla u_i ( q + \omega_i)^T ( z + \omega_i - q - \omega_i)$, for any $z \in \mathbb{R}^\ell$. Specifically, if $z \in Q_i$, then $u_i(z + \omega_i) \geq u_i(\overline{y}_i) > u_i (q + \omega_i)$. Using the previous two inequalities we get the desired separation, via $\pi \coloneqq \nabla u_i ( q + \omega_i)$
\begin{align*}
\pi^T q  & < \pi^T z \qquad \text{ for all } z \in \widehat{Q}_i \subset Q_i.
\end{align*}

In case (ii), the vector $\pi \coloneqq - \frac{q}{\gamma \|q\|} $ suffices; here $\gamma = \sqrt{\frac{2 ( \lambda \ell \delta + 1)}{\alpha}}$. Note that 
\begin{align}
\pi^T q = - \frac{\| q \|}{\gamma} & < - 1 \label{ineq:norm-case}
\end{align}
For any $z \in \widehat{Q}_i$ we have $\| z \| \leq \gamma$. Now, the Cauchy-Schwartz inequality gives us $| \pi^T z |  \leq \| \pi \| \| z \|  = \frac{1}{\gamma} \| z \| \leq 1$. This inequality along with (\ref{ineq:norm-case}) shows that $\pi$ is indeed a separating hyperplane: $\pi^T q < \pi^T z $ for all $z \in \widehat{Q}_i$.


Overall, we observe that separation with respect to the $\widehat{Q}_i$s can be performed efficiently. Hence, via the Ellipsoid method, we obtain, for each $i$, the algorithm $\textsc{Alg}_i$ that optimizes over $\widehat{Q}_i$. 

Combining $\textsc{Alg}_i$s we get the optimization algorithm (over $\widehat{\mathcal{Q}}$) $\textsc{Alg}$, which, in turn, leads to $\textsc{Sep}$ (the desired algorithm that separates with respect to $\widehat{\mathcal{Q}}$). 

\end{proof}

\begin{remark}
The proof of Theorem~\ref{theorem:testing-algorithm} shows that if allocation $\overline{y}$ is an $\varepsilon$-Walrasian allocation, then the hyperplane separating the vector $\nu \coloneqq (-\delta) \one$ from $\mathcal{Q} \coloneqq \cvh \left( \cup_i Q_i \right)$ provides the equilibrium prices $p \in \Delta$. That is, if vector $p \in \Delta$ satisfies $p \cdot \nu \leq p \cdot  \mathcal{Q}$, then $(p, \overline{y})$ is an $\varepsilon$-Walrasian equilibrium. Note that such a price vector can be obtained by considering $\mathcal{P}_{\mathcal{Q}} ( \nu) \in \mathcal{Q}$, the projection (under Euclidean distance) of $\nu$ onto $\mathcal{Q}$. In particular, via the variational characterization of convex projections, we have $(z - \mathcal{P}_{\mathcal{Q}} (\nu) )^T(\nu - \mathcal{P}_{\mathcal{Q}} (\nu) )\leq 0$ for all $z \in \mathcal{Q}$. That is, $ (\mathcal{P}_{\mathcal{Q}} (\nu) - \nu)^T z \geq (\mathcal{P}_{\mathcal{Q}} (\nu) - \nu)^T \mathcal{P}_{\mathcal{Q}} (\nu) \geq (\mathcal{P}_{\mathcal{Q}} (\nu) - \nu)^T \nu$ for all $z \in \mathcal{Q}$. Hence, with $p = \frac{\mathcal{P}_{\mathcal{Q}}(\nu) - \nu}{\| \mathcal{P}_{\mathcal{Q}}(\nu) - \nu \|_1}$, we get the desired separation. 

The norm of $\mathcal{P}_{\mathcal{Q}}(\nu)$ is polynomially bounded: note that $\| \mathcal{P}_{\mathcal{Q}}(\nu) - \nu \| \ \leq \ \min_i \| \overline{y}_i - \nu \|$, since $\overline{y}_i \in \mathcal{Q}$. Another relevant observation is that Lemma~\ref{lemma:bounded-hull} is not confined to $(-\delta)\one$--we can establish such a containment result for any vector $q \in \mathbb{R}^\ell$ as long as we take obtain $\widehat{Q}_i$s  by intersection $Q_i$s with a large enough ball. The radius of the ball just has to be polynomially large in $\| q\|$. Hence, with large enough $\widehat{Q}_i$s, we can ensure that $\mathcal{P}_{\mathcal{Q}}(\nu) \in \widehat{\mathcal{Q}} \coloneqq \cvh  \left( \cup_i \widehat{Q}_i \right) $. 

Additionally, the set containment $ \widehat{\mathcal{Q}} \subset {\mathcal{Q}}$ implies the projection of $\nu$ onto $\widehat{\mathcal{Q}}$, say $\mathcal{P}_{\widehat{\mathcal{Q}}}(\nu) \in \mathbb{R}^\ell$, is the same as the desired vector $\mathcal{P}_{\mathcal{Q}}(\nu)$.
 
In the proof of Theorem~\ref{theorem:testing-algorithm} we have developed a polynomial-time separation oracle for $\widehat{\mathcal{Q}}$. Therefore, via the Ellipsoid method, the projection $\mathcal{P}_{\widehat{\mathcal{Q}}}(\nu)$ can be computed efficiently and, hence, we can find the equilibrium prices $p = \frac{\mathcal{P}_{\widehat{\mathcal{Q}}}(\nu) - \nu}{\| \mathcal{P}_{\widehat{\mathcal{Q}}}(\nu) - \nu \|_1}$.
 \end{remark}

\input{plc}

%% file: plc.tex
\subsection{Testing Algorithm for Economies with Piecewise-Linear Concave Utilities}
\label{section:plc}

In this section we consider economies, $\Econ =((u_i, \omega_i))_{i \in [h]}$, in which consumers' utilities are piecewise-linear concave (PLC). In this PLC setting, for every agent, the utility of each consumption bundle is obtained by taking a minimum over a set of linear functions. Specifically, PLC utilities have the form $u_i(x) \coloneqq \min_{k} \left\{ \sum_{j} U^k_{i,j} x_j \ + \ T^k_i \right\}$,  here $x_j$ is the amount of good $j$ in the consumption bundle $x \in \mathbb{R}^\ell_+$ and the nonnegative parameters $U^k_{i,j} \in \mathbb{R}_+$ and $T^k_i \in \mathbb{R}_+$ define the $k$th linear function for agent $i$. These parameters are given as input to specify each agent's utility. 

While PLC utilities are concave, they are not strongly concave. Hence, under PLC utilities, one cannot directly apply Theorem~\ref{theorem:testing-algorithm}. However, we show that the ideas developed in the previous section can be adapted to obtain a polynomial-time algorithm for testing whether a given allocation is approximately Walrasian in a PLC economy.  Finding equilibria (exact and approximate) is known to computationally hard under PLC utilities~\cite{deng2008computation, chen09hardness,garg2017settling}. Hence, the result in this section identifies an interesting dichotomy between testing and finding a Walrasian equilibrium.   

First, we will establish a containment result analogous to Lemma~\ref{lemma:bounded-hull}. Using this containment result (and arguments similar to the ones developed in Theorem~\ref{theorem:testing-algorithm}) we develop an efficient, testing algorithm for PLC economies in Theorem~\ref{theorem:plc}.

In the PLC context, a key observation (established below in Lemma~\ref{lemma:plc-containment}) is that the containment property can be obtained by considering vectors of norm at most   
 
\begin{align}
\Lambda & \coloneqq \max_{i \in [h], x \in \mathbb{R}^\ell_+ } \left\{ \| x - \omega_i \| \ : \  u_i (x) \leq u_i \left( \sum_i \omega_i \right)  \right\} \label{eq:defn-Lambda}
\end{align}

Under PLC utilities, the bit complexity of $\Lambda$ is polynomially bounded: specifically, for any consumption bundle $x \in \mathbb{R}^\ell_+$, with the property $u_i (x) \leq u_i \left( \sum_i \omega_i \right)$, we have, for each component $a \in [\ell]$, $x_a \leq \frac{u_i \left( \sum_i \omega_i \right)}{ \min_k U^k_{i,a}} \leq \frac{u_i \left( \sum_i \omega_i \right)}{\min_{k,j} U^k_{i,j}}$.\footnote{We can address the case in which one of the nonnegative coefficients $U^k_{i,j}$ is equal to zero. In particular, adapting the arguments in Lemma~\ref{lemma:plc-containment} one can obtain $x_a \leq \frac{u_i \left( \sum_i \omega_i \right)}{\min_{\{ k, j: U^k_{i,j} \neq 0\} } \ U^k_{i,j} }$.} The bit complexity of this upper bound is polynomially large and, hence, $\Lambda$ is also polynomially upper bounded, in terms of bit complexity. 

Note that, this bound does not require the utilities to be normalized. In fact, the results developed in this section hold for any exchange economy wherein the utilities are concave and $\Lambda$ (as defined in \eqref{eq:defn-Lambda}) is appropriately bounded.

For any given allocation $\overline{y} = (\overline{y}_i)_{i \in [h]}$, consider the bundle $\overline{y}_i \in \mathbb{R}^\ell$ allocated to consumer $i$ and, as before, write $Q_i \coloneqq \{ z \in \mathbb{R}^\ell \mid  z + \omega_i \in \mathbb{R}^\ell_+ \text{ and } u_i ( z + \omega_i) \geq u_i (\overline{y}_i) \}$. In addition, we define a bounded subset of $Q_i$ 
\begin{align}
\widetilde{Q}_i & \coloneqq Q_i \cap \left\{ z \in \mathbb{R}^\ell \mid  \ \| z \| \leq \Lambda \right\} \label{eq:defn-tildeQ}
\end{align}

For each consumer $i$, the subset $\widetilde{Q}_i$ is compact, convex, and has a nonempty interior. 

\begin{lemma}
\label{lemma:plc-containment}
Let  $\overline{y}$ be an allocation in an exchange economy $\Econ$ with PLC utilities. Suppose that the sets $Q_i$ and $\widetilde{Q}_i$, for $i\in [h]$, are as defined above. Then, with parameter $\da>0$, we have 
\begin{align*}
(-\delta) \one \in \cvh \left( \bigcup_{i=1}^h Q_i\right) & \ \ \ \textrm{  iff  } \ \ \ (-\delta) \one \in \cvh \left( \bigcup_{i=1}^h \widetilde{Q}_i\right).
\end{align*}
\end{lemma}
\begin{proof}
The proof of this claim is almost identical to that of Lemma~\ref{lemma:bounded-hull}--the difference being that here we use the bound provided by $\Lambda$, instead of relying on strong concavity.

To begin with, note that the reverse direction of the claim is direct, since $\widetilde{Q}_i \subset Q_i$ for all $i \in [h]$. 

For the forward direction, we have vectors $z_i \in Q_i$ and a convex combination $\lambda_i \geq 0$, for $i \in [h]$, such that $\sum_{i=1}^h \lambda_i = 1$ and
\begin{align}
\sum_{i=1}^h \lambda_i z_i = (-\delta) \one \label{eq:comb-repeat}
\end{align}

Let $R :=\max_i \{\| z_i \| :i\in [h]\}$. By definition, the $z_i$s are contained in the (closed) Euclidean ball $B(R)$ of radius $R$ and center $\zero$. 
Note that, for each $i \in [h]$, the intersection $Q_i \cap B(R)$ is a compact set. 

Let $Z$ denote the collection of all tuples \[(z'_1, z'_2, \ldots, z'_h) \in \left(Q_1 \cap B(R) \right) \times \left( Q_2 \cap B(R) \right) \times \ldots \times \left( Q_h \cap B(R)\right)\] for which there exists exists convex coefficients $\lambda'_i$s such that $\sum_i \lambda'_i z'_i \leq (-\delta) \one$, i.e., there exists a convex combination of $z'_i$s which is component-wise upper bounded by $(-\delta) \one$.\footnote{The definition of $Q_i$s provide a component-wise lower bound as well: $z_i \geq - \omega_i$, for each vector $z_i \in Q_i$.}

From~\eqref{eq:comb-repeat}, we know that $Z$ is nonempty. Given that the sets $(Q_i \cap B(R))$s are compact, one can show that $Z$ is compact as well. 
Hence, the problem of minimizing $\max\left\{ \|z_i' -  (-\delta) \one \| \ \mid (z'_i)_i \in Z \right\}$ admits an optimal solution, say $(z^*_t)_t$. Note that, by definition of $Z$, there exists convex coefficients $(\lambda^*_t)_{t \in H^*}$ that satisfy $\sum_{ t \in H^*} \lambda^*_t z^*_t \leq (-\delta) \one$;  
here, subset $H^* \subseteq [h]$ is selected to ensure that $\lambda^*_t > 0 $ for all $t \in H^*$. 
 
Next, we will prove that $\| z^*_t \| \leq \Lambda $, for all $t \in H^*$. Subsequently, we will show that using $z^*_t$s we can obtain vectors $\widetilde{z}_t \in Q_t$ that satisfy the same norm bound ($\| \widetilde{z}_t \| \leq \Lambda $) and whose convex combination is equal to $(-\delta) \one$. This norm bound implies that $\widetilde{z}_t \in \widetilde{Q}_t$ and, hence, leads to the desired containment: $(-\delta) \one \in \cvh \left( \bigcup_{i=1}^h \widetilde{Q_i} \right)$. 
 
The bundle $x^*_i \coloneqq z^*_i + \omega_i$ must satisfy $u_i(x^*_i)=u_i(\bar y_i) \leq u_i(\sum_i \omega_i)$: if this is not the case (i.e., we have a strict inequality $u_i(x^*_i) > u_i(\bar y_i)$), then by reducing a positive component\footnote{If all the components of $z^*_i$ are negative, then $x^*_i \leq \omega_i$ and we get the desired bound $u_i(x^*_i) \leq u_i(\omega_i) \leq u_i(\sum_i \omega_i) $.} of $z^*_i = x^*_i - \omega_i$ we can ensure that $z^*_i$ moves closer to $(-\delta) \one$ and at the same time $z^*_i$ continues to be in $Q_i$ (i.e., the following inequality continues to hold $u_i(x^*_i) \geq u_i(\bar y_i)$). Also, note that such a reduction maintains the containment of $z^*_i$s in $Z$; specifically, the following inequality continues to hold $\sum_{ t } \lambda^*_t z^*_t \leq (-\delta) \one$. A repeated application of this argument gives us $x^*_i = z^*_i + \omega_i$ with the property that $u_i(x^*_i)=u_i(\bar y_i) \leq u_i \left( \sum_i \omega_i \right)$, for all $i$. Therefore, using the definition of $\Lambda$ (see \eqref{eq:defn-Lambda}), we obtain the stated bound $\| z^*_i \| \leq \Lambda$. 

To complete the proof we will show that $z^*_t$s can be transformed into vectors $\widetilde{z}_t \in Q_t$ that satisfy the same norm bound and whose convex combination is equal to $(-\delta) \one$. Write $\varphi \coloneqq \sum_t \lambda^*_t z^*_t $ and note that $\varphi \leq (-\delta) \one$. If component $a \in [\ell]$ of $\varphi$ is strictly less than $-\delta$, then there exists a $z^*_i$ such that its $a$th component is also less than $-\delta$: $z^*_{i, a} < -\delta$. We can increase $z^*_{i,a}$ till either it becomes equal to zero, or the $a$th component of $\varphi$ reaches $-\delta$.\footnote{Here, the convex coefficients, $\lambda^*_t$s, remain unchanged.} Note that in this transformation while the $a$th component of $z^*_i$ increases in value, it decreases in magnitude. Hence, the utility $u_i(z^*_i + \omega_i)$ increases and the norm of $z^*_i$ decreases. Repeatedly applying this procedure gives us vectors $\widetilde{z}_t \in Q_t$ such that $\| \widetilde{z}_t \| \leq \Lambda$ and $\sum_t \lambda^*_t \widetilde{z}_t  = (-\delta) \one$. 
 
Overall, this implies that $(-\delta) \one \in \cvh \left( \bigcup_{i=1}^h \widetilde{Q}_i\right)$ and the stated claim follows. 
\end{proof}

Lemma~\ref{lemma:plc-containment} leads to the following characterization for approximate equilibria in PLC economies. 

\begin{lemma}
\label{lemma:non-containment-plc}
An allocation $\overline{y}$ is an $\varepsilon$-Walrasian allocation in a PLC economy $\Econ$ iff 
\begin{align*}
\left(-\delta \right) \one \notin \cvh \left( \bigcup_{i=1}^h \widetilde{Q}_i \right).
\end{align*}
Here, for each consumer $i \in [h]$, the set $\widetilde{Q}_i$ is as defined above.
\end{lemma}

The proof of this result is identical to that of Lemma~\ref{lemma:non-containment} and is omitted; one has to simply use Lemma~\ref{lemma:plc-containment}, instead of Lemma~\ref{lemma:bounded-hull}.

We now establish the main result of this section. 
 
\begin{theorem}
\label{theorem:plc}
There exists a polynomial-time algorithm that---given an allocation $\overline{y} = (\overline{y}_i)_{i \in [n]}$ in an exchange economy $\Econ=((u_i, \omega_i))_{i \in [n]}$ with PLC utilities---determines whether $\overline{y}$ is an $\varepsilon$-Walrasian allocation, or not.
\end{theorem}
\begin{proof}
The design of the testing algorithm for PLC utilities is quite similar to the method developed in the proof of Theorem~\ref{theorem:testing-algorithm}. We present the details for the PLC setting for completeness. 

As a direct consequence of Lemma~\ref{lemma:non-containment-plc} we have: testing for approximately Walrasian allocation corresponds to determining whether the vector $(-\delta) \one$ is contained in the convex hull of the $\widetilde{Q}_i$s; see~\eqref{eq:defn-tildeQ} for the definition of these sets. 

Write $\widetilde{\mathcal{Q}} \coloneqq \cvh \left( \bigcup_{i=1}^h \widetilde{Q}_i \right)$. We will develop an efficient algorithm, \textsc{Alg}, for solving linear optimization problems over $\widetilde{\mathcal{Q}}$, i.e., for solving problems of the form 
\begin{align} 
\max \ c^Tz & \ \ \  \text{ subject to } \ z \in \widetilde{\mathcal{Q}} \label{eq:LP-plc}
\end{align}
Here, $c \in \mathbb{R}^\ell$ is an input vector. The equivalence of optimization and separation (see, e.g.,~\cite{grotschel2012geometric}) implies that \textsc{ALG} can be used to design a polynomial-time algorithm \textsc{Sep} that provides a separation oracle for $\widetilde{\mathcal{Q}}$. That is, using \textsc{Sep} we can perform the desired test of determining whether $(-\delta) \one \in \widetilde{\mathcal{Q}}$, or not. 

In order to apply the optimization-separation equivalence we need to ensure that $\widetilde{\mathcal{Q}}$ is compact, convex, and has a nonempty interior. These properties are satisfied by $\widetilde{Q}_i$s individually, hence they hold for $\widetilde{\mathcal{Q}}$ as well. Therefore, we can evoke the equivalence (via an application of the Ellipsoid method over the \emph{polar} of $\widetilde{\mathcal{Q}}$) and obtain the algorithm \textsc{Sep}. 

To develop the algorithm, \textsc{Alg}, that efficiently solves linear optimization problems of the form (\ref{eq:LP-plc}), we note that the feasible set $\widetilde{\mathcal{Q}}$ is a convex hull of the $\widetilde{Q}_i$s. Hence, for any $c \in \mathbb{R}^\ell$, an optimal solution of (\ref{eq:LP-plc}) can be obtained by solving 
\begin{align}
\max_{i \in [h]} \left( \max \ c^Tz_i  \ \ \  \text{ subject to } \ z_i \in \widetilde{Q}_i \right) \label{eq:LP1-plc}
\end{align}
Here, for each $i$, the decision variable $z_i \in \mathbb{R}^\ell$ lies in the set $\widetilde{Q}_i$. Below we will provide, for each $i$, a polynomial-time algorithm, $\textsc{Alg}_i$, that solves the linear optimization problem over $\widetilde{Q}_i$, i.e., $\textsc{Alg}_i$ efficiently solves $\max \ c^Tz_i  \ \  \text{ subject to } \ z_i \in \widetilde{Q}_i$. Hence, $\textsc{Alg}$ can be obtained by directly taking a maximum over the (optimal) solutions obtained by the $\textsc{Alg}_i$s. 

We will now complete the chain of arguments mentioned above by designing the optimization algorithm $\textsc{Alg}_i$. This algorithm is itself based on the Ellipsoid method. As detailed below, the gradients of the utility function $u_i$ (at different points) can be used to separate  $\widetilde{Q}_i$ from vectors that are not contained in it. Hence, with this separation technique in hand, we can apply the Ellipsoid method over $\widetilde{Q}_i$ to obtain $\textsc{Alg}_i$.\footnote{Recall that, $\widetilde{Q}_i$s are compact, convex, and have a nonempty interior. Hence, the Ellipsoid method is applicable over these sets.}  	

Given a query vector $q \in \mathbb{R}^\ell$, it is easy to test if $q \in \widetilde{Q}_i  \coloneqq Q_i \cap \ \left\{ z \in \mathbb{R}^\ell \ : \ \| z \| \leq  \Lambda \right\} $. We directly verify (a) $u_i(q + \omega_i) \geq u_i(\overline{y}_i)$  (to ensure that $q \in Q_i$) and (b) $\| q \| \leq \Lambda$.\footnote{We do not have to compute $\Lambda$ exactly. Here, for the algorithm, an upper bound (with polynomial bit complexity) of $\Lambda$ suffices.}  

Consider the case in which $q \notin \widetilde{Q}_i $. To run the Ellipsoid method (that underlies $\textsc{Alg}_i$), we need a separating hyperplane for such a $q \in \mathbb{R}^\ell$. There are two complementary (though, nonexclusive) cases either (i) $q \notin Q_i$ (i.e., $u_i(q + \omega_i) < u_i(\overline{y}_i)$) and (ii) $\| q \| > \Lambda$.

In case (i), the subgradient at $\chi \coloneqq q + \omega_i$ provides the separating hyperplane: in particular, for the PLC utility $u_i$, let $\bar{k} \in \argmin_{k} \left\{ \sum_{j} U^k_{i,j} \ \chi_j \ + \ T^k_i \right\}$, i.e., $u_i(q + \omega_i) = \sum_{j} U^{\bar{k}}_{i,j} \chi_j \ +  T^{\bar{k}}_i$. The (subgradient) vector $\pi \coloneqq (U^{\bar{k}}_{i,j})_{j \in [\ell]}$ provides the desired separation. This follows from the fact that, for any $z \in \mathbb{R}^\ell$ we have 
\begin{align}
u_i(z + \omega_i) & \leq  \pi^T \left( z + \omega_i \right) + T^{\bar{k}}_i \nonumber \\ 
& =  \pi^T \left( q + \omega_i \right) + T^{\bar{k}}_i + \pi^T \left( z  - q \right) \nonumber \\
& = u_i( q + \omega_i) + \pi^T \left( z  - q \right) \label{ineq:subgradient} 
\end{align}

Specifically, if $z \in Q_i$, then $u_i(z + \omega_i) \geq u_i(\overline{y}_i) > u_i (q + \omega_i)$. Hence, using \eqref{ineq:subgradient}, we get the desired separation
\begin{align*}
\pi^T q  & < \pi^T z \qquad \text{ for all } z \in \widetilde{Q}_i \subset Q_i.
\end{align*}

In case (ii), the vector $\pi \coloneqq - \frac{q}{\Lambda \|q\|} $ suffices. Note that 
\begin{align}
\pi^T q = - \frac{\| q \|}{\Lambda} & < - 1 \label{ineq:norm-case}
\end{align}
For any $z \in \widetilde{Q}_i$ we have $\| z \| \leq \Lambda$. Now, the Cauchy-Schwartz inequality gives us $| \pi^T z |  \leq \| \pi \| \| z \|  = \frac{1}{\Lambda} \| z \| \leq 1$. This inequality along with (\ref{ineq:norm-case}) shows that $\pi$ is indeed a separating hyperplane: $\pi^T q < \pi^T z $ for all $z \in \widetilde{Q}_i$.


Overall, we observe that separation with respect to the $\widetilde{Q}_i$s can be performed efficiently. Hence, via the Ellipsoid method, we obtain, for each $i$, the algorithm $\textsc{Alg}_i$ that optimizes over $\widetilde{Q}_i$. 

Combining $\textsc{Alg}_i$s we get the optimization algorithm (over $\widetilde{\mathcal{Q}}$) $\textsc{Alg}$, which, in turn, leads to $\textsc{Sep}$ (the desired algorithm that separates with respect to $\widetilde{\mathcal{Q}}$).

\end{proof}

%% file: appendix.tex
\appendix
\section{Other Notions of Approximate Equilibria}
\label{appendix:other-notions}
In this section we will show that the notion of approximate Walrasian equilibrium considered in this work relates to other notions studied in computer science and economics. 

A  common alternate (to the one considered in this work) definition of approximation seeks to relax how exactly the consumers optimize, see, e.g.,~\cite{deng2002complexity} and~\cite{garg2017settling}. This notion requires that the consumers are approximately maximizing their utilities, though it assumes that the demand is approximately equal to the supply. Specifically, under this complementary notion, a pair $(p, \overline{x}) \in \Delta \times \mathbb{R}^{h \ell}_+$, with price vector $p \in \Delta$ and allocation $\overline{x}$, is said to be an $\widehat{\varepsilon}$-equilibrium (in an economy $\Econ=((u_i, \omega_i))_i$ with $h$ consumers and $\ell$ goods) iff the following conditions hold for all consumers $i \in [h]$:\footnote{Note that (C1) and (C2) are a strengthening of the alternate formulation mentioned above, since here we insist that the demand is exactly equal to the supply.}

\begin{enumerate}
\item[(C1)] For any bundle $y \in \mathbb{R}^\ell_+$, with the property that $p^T y \leq p^T \omega_i$, we have $u_i(y) \leq u_i(\overline{x}_i) + \widehat{\varepsilon}$.
\item[(C2)] $|p^T \overline{x}_i - p^T \omega_i | \leq \widehat{\varepsilon}$.
\end{enumerate}

Note that, (C1) corresponds to the following requirement: $u_i(\overline{x}_i) \geq \left(\max_{y \in \mathbb{R}^\ell_+} \left\{ u_i(y) \mid p^T y \leq p^T \omega_i \right\} \right)- \widehat{\varepsilon}$.

We will show that an ${\varepsilon}$-approximate Walrasian equilibrium (as defined in Section~\ref{section:notation}) satisfies conditions (C1) and (C2) with $\widehat{\varepsilon} = \frac{\varepsilon \lambda \sqrt{\ell}}{h}$; here, $\lambda$ is the Lipschitz constant of the utilities $u_i$s.  

By definition (see Section~\ref{section:notation}), an $\varepsilon$-approximate Walrasian equilibrium, say $(p, \overline{x})$, satisfies (C2).\footnote{Here, assume that $\widehat{\varepsilon} \geq \varepsilon$, otherwise we can consider $\widehat{\varepsilon} = \max \left\{ \frac{ \varepsilon  \lambda \sqrt{\ell}}{h} , \varepsilon \right\}$.} 

To establish the claim, next we will prove the contrapositive form of (C1) holds for $(p, \overline{x})$. Consider a bundle $y \in \mathbb{R}^\ell_+$ with high utility $u_i(y) > u_i(\overline{x}_i) + \widehat{\varepsilon} = u_i(\overline{x}_i) + \frac{\varepsilon \lambda \sqrt{\ell}}{h} $. 

Observe that the Euclidean ball $B\left(y, \frac{ \varepsilon \sqrt{\ell}}{ h} \right)$, with center $y$ and radius $\frac{ \varepsilon \sqrt{\ell}}{ h}$, is entirely contained in the upper contour set $\left\{ x \in \mathbb{R}^\ell_+ \mid u_i(x) \geq u_i(\overline{x}_i) \right\}$. Say, towards a contradiction, that this is not the case. Then, there exists a bundle $y' \in \mathbb{R}^\ell_+$ on the boundary of the upper contour set that is at most a distance $\frac{\varepsilon \sqrt{\ell}}{h}$ away from $y$, i.e., a boundary point $y'$ that satisfies $\| y - y' \| \leq \frac{ \varepsilon \sqrt{\ell}}{ h}$. Since $y'$ is at the boundary of the upper contour set, $u_i(y') = u_i(\overline{x}_i)$. This fact contradicts the high utility of $y$: $|u_i(y) - u_i(y') | \leq \lambda \| y - y '\| \leq \frac{\varepsilon \lambda \sqrt{\ell}}{h}$.

Therefore, in particular, the vector $\left( y - \frac{\varepsilon \sqrt{\ell}}{h} \frac{p}{\| p \|} \right) \in B\left(y, \frac{\varepsilon \sqrt{\ell}}{h}  \right)$ belongs to the upper contour set. That is, $u_i \left( y - \frac{\varepsilon \sqrt{\ell}}{h} \frac{p}{\| p \|} \right) \geq u_i(\overline{x}_i)$. 

Hence, our definition of $\varepsilon$-Walrasian equilibria ensures that 
\begin{align}
p^T \left( y - \frac{\varepsilon \sqrt{\ell}}{h} \frac{p}{\| p \|} \right) \geq p^T \omega_i - \varepsilon/h \label{ineq:deficit}
\end{align}
Since the price vector $p \in \Delta$, we have $\| p \| \geq \frac{1}{\sqrt{\ell}}$. Rearranging inequality (\ref{ineq:deficit}) shows that $y$ satisfies (the contrapositive form of) the condition (C1): $p^T y \geq p^T \omega_i$. \\

\noindent
{\bf Supply vs demand and average budge gap.} 
In contrast to some of the other formulations, the definition of approximate equilibria considered in this work strictly enforces that the supply is equal to the demand. Also, we consider a bound on the budget gap, $|p^T \overline{x}_i - p^T \omega_i|$, for every consumer $i$. Hence, our approximate equilibrium satisfies the average budget gap requirement considered in~\cite{mas1979refinement}. 

\section{Illustrative Example of Strongly Concave Utilities}
\label{appendix:example}
This section provides a simple example to illustrate the interplay of strong concavity and other parameters related to the utility functions. Consider an economy $\Econ=((u_i, \omega_i))_{i \in [h]}$ (with $h$ consumers and $\ell$ goods) wherein the utility of each consumer is 
\begin{align*}
u(x) \coloneqq \frac{1}{N} \sum_{j=1}^\ell  \sqrt{x_j  + \theta }.
\end{align*}
Here, $x_j$ is the amount of good $j$ (present in the consumption bundle $x \in \mathbb{R}^\ell_+$), $\theta >0$ is a fixed constant, and $N >0$ is a normalization term. Here, dividing by a large enough, but fixed, parameter $N$ ensures that for each feasible bundle $x \in \mathbb{R}^\ell_+$ in $\Econ$, we have $u(x) \in (0,1)$. In particular, let $\overline{\omega}_j$ denote the total amount of good $j \in [\ell]$ present in the economy ($\overline{\omega}_j\coloneqq \sum_i \omega_{i,j}$). Then, $N = \sum_{j=1}^{\ell} \sqrt{\overline{\omega}_j  + \theta }$.

Also, note that the additive shift of $\theta$ ensures that the gradient of $u$ is bounded for all $x \in \mathbb{R}^\ell_+$ and the Lipschitz constant of the utilities $\lambda = \frac{1}{2\sqrt{\theta}}$.

One can show that a function $u: \mathbb{R}^\ell \mapsto \mathbb{R}$ is $\alpha$-strongly concave in a set $\mathcal{R} \subset \mathbb{R}^\ell$ iff $f(x) \coloneqq u(x) + \frac{\alpha}{2} \| x \|^2$ is concave within $\mathcal{R}$. That is, for any $\alpha >0$, if the Hessian of $f(x)$ is negative semidefinite for all $x \in \mathcal{R}$, then $u$ is $\alpha$ strongly concave in $\mathcal{R}$. 

The utility function $u$ is separable across the $\ell$ goods, hence its Hessian (at point $x \in \mathbb{R}^\ell_+$) is a diagonal matrix with the (diagonal) entries being $\frac{-1}{4 N ( x_j + \theta)^{3/2}}$ for each $j \in [\ell]$. That is, for an $\alpha >0$ and $x \in \mathbb{R}^\ell_+$, the Hessian of $f$ is again a diagonal matrix with entries $\frac{-1}{4 N ( x_j + \theta)^{3/2}} + \alpha$, with $ j \in [\ell]$. 

Therefore, within the Euclidean ball of radius $r$ (and center $\zero$), the function $u$ is  $\left(\frac{1}{4 N ( r + \theta)^{3/2}}\right)$-strongly concave; this value of $\alpha$ ensures that $f$ is negative semidefinite throughout the ball. 
  
 As stated in Section~\ref{section:notation}, we require strongly concavity to hold within an appropriately large set around the endowments.\footnote{For ease of presentation, here we assume that the endowments are $\zero$. Shifting all the vectors by $-\omega_i$, directly provide  the arguments for general endowment vectors.} Specially, we need $\alpha r^2 \geq \frac{2\varepsilon \lambda \ell}{h} + 2$ (see Section~\ref{section:notation}). Using the above-mentioned expression for the modulus of strong concavity, $\alpha$, the required inequality translates to 
 \begin{align*}
 \frac{\sqrt{r}}{4 N (1 + \theta/r)^{3/2}} \geq \frac{2\varepsilon \lambda \ell}{h} + 2
 \end{align*}   
 
 The left-hand-side of the previous inequality is an increasing function of $r$. Hence, this inequality holds for an appropriately large $r$ (which only depends on the parameters of $\Econ$ and not on the replication factor $n$). With this $r$ and the corresponding value of $\alpha$ in hand one can apply the results developed in this work.

\section{Proof of Lemma~\ref{lem:equaltreatment}}
\label{appendix:equal-treatment}

This section shows that, if the utilities of the consumers are strictly monotonic, continuous, and strictly concave, then allocations in the $h$-core of a replica economy satisfy the equal treatment property.

Let $\overline{x} = (\overline{x}_{1,1}, \ldots, \overline{x}_{1,h}, \ldots, \overline{x}_{n,1}, \ldots, \overline{x}_{n,h})$ be an $h$-core allocation of the economy $\Econ^n$. We will show that if $\overline{x}$ does not satisfy the equal treatment property, then a coalition of size $h$---consisting of the worse off consumers of each type---will block the allocation $\overline{x}$, contradicting the fact that it is in the $h$-core.  

Say, towards a contradiction, that $\overline{x}$ does not satisfy the equal treatment property. Since each consumer of type $t \in [h]$ has the same utility function, we can select the worse off consumer among the $n$ ones that have type $t$; in particular, for each $t \in [h]$, let index $j^*(t) \in [n]$ be such that $u_t (\overline{x}_{j^*(t), t} ) \leq u_t(\overline{x}_{i,t})$ for all $i \in [n]$. 

Consider the size-$h$ coalition $S \coloneqq \{(j^*(t), t)\}_{t \in [h]}$ and the $S$-allocation $(\widehat{x}_{t})_{(t \in [h])}$ defined as 
$\widehat{x}_{t}  \coloneqq  \frac{1}{n} \sum_{i=1}^n \overline{x}_{i, t}$.

Note that $(\widehat{x}_{t})_{t \in [h]}$ is indeed an $S$-allocation: Since $\overline{x}$ is an allocation in $\Econ^n$, it satisfies $\sum_{i=1}^n \sum_{t=1}^h \overline{x}_{i,t} = \sum_{i=1}^n \sum_{t=1}^h \omega_t = \sum_{t=1}^h \ n \omega_t$. Dividing by $n$ gives us $\sum_{t=1}^h \widehat{x}_t = \sum_{t=1}^h \omega_t$. Therefore, the consumers in the coalition $S$ can trade among themselves and, individually, obtain bundles $\widehat{x}_{t}$s. 

Furthermore, the strict concavity of $u_t$s ensures that $u_t(\widehat{x}_t) \geq u_t(\overline{x}_{j^*(t), t})$ for all $(j^*(t), t) \in S$---with one of the inequalities being strict. Hence, coalition $S$ blocks the allocation $\overline{x}$ and, by way of contradiction, the stated claim follows.